\documentclass[aps, pra, showpacs, superscriptaddress, 10pt, twocolumn]{revtex4-1}

\usepackage{amsfonts,amssymb,amsmath,amsthm,longtable,array}
\usepackage{graphics,graphicx,color}
\usepackage{epstopdf,epsfig}
\newcommand{\ket}[1]{ | #1 \rangle}
\newcommand{\bra}[1]{ \langle #1  |}

\newcommand{\ketbra}[1]{ | #1 \rangle \langle #1  |}

\newcommand{\beq}{\begin{equation}}
\newcommand{\eeq}{\end{equation}}
\newcommand{\ben}{\begin{eqnarray}}
\newcommand{\een}{\end{eqnarray}}
\usepackage{braket}
\usepackage{todonotes}
\usepackage{natbib}
\usepackage[shortlabels]{enumitem}

\newcommand\oa[1]{\text{OA}\left( #1  \right) }

\newtheorem{lemma}{Lemma}
\newtheorem{observation}{Observation}
\newtheorem{proposition}{Proposition}

\newtheorem{corollary}{Corollary}
\newtheorem{conjecture}{Conjecture}

\theoremstyle{definition}
\newtheorem{definition}{Definition}
\newtheorem{example}{Example}
\newtheorem{remark}{Remark}

\DeclareMathOperator{\tr}{tr}
\DeclareMathOperator{\Id}{Id}

\usepackage{braket}

\usepackage{tikz}
\usetikzlibrary{fit}
\tikzset{%
  highlight/.style={rectangle,rounded corners,fill=red!15,draw,fill opacity=0.5,thick,inner sep=0pt}
}
\newcommand{\tikzmark}[2]{\tikz[overlay,remember picture,baseline=(#1.base)] \node (#1) {#2};}

\tikzset{%
  highlight/.style={rectangle,rounded corners,fill=red!15,draw,fill opacity=0.5,thick,inner sep=0pt}
}
\newcommand{\Highlight}[1][submatrix]{%
    \tikz[overlay,remember picture]{
    \node[highlight,fit=(left.north west) (right.south east)] (#1) {};}
}
\newcommand{\tikzmarkk}[2]{\tikz[overlay,remember picture,baseline=(#1.base)] \node (#1) {#2};}

\tikzset{%
  highlightt/.style={rectangle,rounded corners,fill=blue!15,draw,fill opacity=0.5,thick,inner sep=0pt}
}
\newcommand{\Highlightt}[1][submatrix]{%
    \tikz[overlay,remember picture]{
    \node[highlightt,fit=(up.north west) (down.south east)] (#1) {};}
}

\usetikzlibrary[positioning]
\usetikzlibrary{patterns}

\usetikzlibrary{shapes.arrows}
\usepackage{centernot}
\usepackage{mathtools}
\usepackage{ stmaryrd }

\newcommand\y{\cellcolor{green!10}}

\usepackage{xparse}
\usetikzlibrary{matrix,backgrounds}
\pgfdeclarelayer{myback}
\pgfsetlayers{myback,background,main}
\tikzset{mycolor/.style = {line width=1bp,color=#1}}%
\tikzset{myfillcolor/.style = {draw,fill=#1}}%

\NewDocumentCommand{\highlight}{O{blue!40} m m}{%
\draw[mycolor=#1] (#2.north west)rectangle (#3.south east);
}

\NewDocumentCommand{\fhighlight}{O{blue!40} m m}{%
\draw[myfillcolor=#1] (#2.north west)rectangle (#3.south east);
}

\PassOptionsToPackage{breaklinks}{hyperref}
\usepackage[colorlinks,allcolors=black]{hyperref}

\usepackage{cleveref}
\crefname{proposition}{Proposition}{Propositions}
\crefname{definition}{Definition}{Definitions}
\crefname{lemma}{Lemma}{Lemmas}
\crefname{figure}{Figure}{Figures}
\crefname{corollary}{Corollary}{Corollary}
\crefname{conjecture}{Conjecture}{Conjectures}
\crefname{section}{Section}{Sections}
\crefname{appendix}{Appendix}{Appendixes}
\crefname{observation}{Observation}{Observation}
\crefname{remark}{Remark}{Remark}
\crefname{example}{Example}{Examples}
\crefname{equation}{Eq.}{Eqs.}
\crefname{table}{Table}{Tables}

\begin{document}
\setcitestyle{numbers,acm}
\title{Stochastic Local Operations with Classical Communication of \\ Absolutely Maximally Entangled States}

\author{Adam Burchardt}
\affiliation{Institute of Physics, Jagiellonian University, \L{}ojasiewicza 11,
	30-348 Krak\'ow, Poland}
	
\author{Zahra Raissi}
\affiliation{ICFO-Institut de Ciencies Fotoniques, The Barcelona Institute of Science and Technology, 08860 Castelldefels (Barcelona), Spain}
	

\begin{abstract}
Absolutely Maximally Entangled (AME) states are maximally entangled for every bipartition of the system. They are crucial resources for various quantum information protocols. 
We present techniques for verifying either two AME states are equivalent concerning Stochastic Local Operations and Classical Communication (SLOCC). 
The conjecture that for a given multipartite quantum system all AME states are SLOCC-equivalent is falsified. 
We also show that the existence of AME states with minimal support of 6 or more particles results in the existence of infinitely many such non-SLOCC-equivalent states. 
Moreover, we present AME states which are not SLOCC-equivalent to the existing AME states with minimal support.
\end{abstract}
	
\maketitle
	
\section{Introduction}

Entanglement of bipartite states is a widely discussed problem and, in fact, already very well understood \cite{PhysRevA.67.022320}. 
However, quantification of entanglement for multipartite states remains a challenge \cite{LUupto5qubits}. 
In particular, according to different entanglement measures for multipartite states  (like the tangle, the Schmidt measure, the localizable entanglement, or geometric measure of entanglement), the states with the largest entanglement do not overlap in general.

We discuss states known as Absolutely Maximally Entangled (AME), which are maximally entangled for every bipartition of the system. 
AME states are being applied in several branches of quantum information theory: in quantum secret sharing protocols \cite{HelwigAME}, in parallel open-destination teleportation \cite{helwig2013absolutely}, in holographic quantum error correcting codes \cite{Pastawski2015HolographicQE}, among many others. 
Different families of AME states have been introduced \cite{Rains1999NonbinaryQC,Helwig2013AbsolutelyME} and the problem of their existence is being investigated \cite{PhysRevA.69.052330,HIGUCHI2000213,Felix72}. 
It has been demonstrated that the simplest class of AME states, namely AME states with the \emph{minimal support}, is in one-to-one correspondence with the classical error correction codes \cite{AME-QECC-Zahra} and combinatorial designs known as \emph{Orthogonal Arrays} \cite{ComDesi}. 
Henceforward, two-way interaction with combinatorial designs and quantum error correction codes is observed \cite{PhysRevA.69.052330,DiK}. 
AME states are special cases of \emph{k-uniform} states characterized by the property that all of their reductions to k parties are maximally mixed \cite{kUNI}.



Since the entanglement quantification becomes an ambitious project while the number of parties in a system increases, the problem of satisfactory classification of states turns to be essential \cite{LUupto5qubits}. 
The state space might be partitioned into equivalence classes 
with respect to 
selected class of local operations \cite{PhysRevLett.83.436}. 
Any two states from one class are interconvertible by an adequate local operator, while such a transformation cannot be provided for states from different classes. 
Nevertheless, it is not obvious which class of local operators provides the ultimate division of state space. 
One of the reasonable choices was the division according to Local Unitary operations (LU) \cite{LUupto5qubits}. 
Two states $\ket{\psi}$ and $\ket{\phi}$ belong to the same LU class iff there exists a local unitary operator transforming one into the other: 
$\Ket{\phi} = U_1 \otimes \cdots \otimes U_N \Ket{\psi}$.

The fact that entanglement is used for the transmission of information between parties far apart, restrict us to the LU operations. 
Nevertheless, we may also allow classical information to be transmitted between the distant parties. 
This leads us to supersede the class of LU operations with the Local Operators and Classical Communications (LOCC) \cite{EverythingYouAlwaysWantedtoKnowAboutLOCC,LocalQuantumTransformationsRequiringInfiniteRoundsofClassicalCommunication}. 
It is known that if the state $\ket{\psi}$ can be transformed into $\ket{\phi}$ by using LOCC operations only, $\ket{\psi}$ possesses at least as much entanglement as $\ket{\phi}$. 
In general, this transformation cannot be inverted, and hence LOCC imposes the partial order on the state space. 
Nevertheless, one may study whether two states are LOCC equivalent with a non-vanishing probability of success. 
Such operations are known as stochastic LOCC (SLOCC).
Mathematically, two states $\ket{\psi}$ and $\ket{\phi}$ belong to the same SLOCC class iff there exists a local invertible operator transforming one into the other: 
$\Ket{\phi} = O_1 \otimes \cdots \otimes O_N \Ket{\psi}$ \cite{ThreeQub}. 
Thus, the state space might be partitioned into SLOCC-classes.

The number of SLOCC-classes rapidly increases with the number of parties $N$ and the local dimension $d$ of a given system. 
For instance, all bipartite pure states are equivalent by SLOCC \cite{PhysRevLett.83.436}. 
In the simplest multipartite system, namely three qubits system, there are precisely six distinct SLOCC classes (only two of them among fully entangled states, represented by GHZ and W states respectively) \cite{ThreeQub}. Those systems are the last ones with a finite number of SLOCC classes, for $N>3$ or $d>2$ there are infinitely many SLOCC classes \cite{ThreeQub}. 
Despite of this fact, all 4-partite qubit sates were classified into nine families, some of them with infinitely many SLOCC classes but of a similar structure \cite{FourQubits}. 
This result was later corrected to eight such classes, while one of the proposed families turned out to be not fully entangled \cite{FourQubits8}. 

For larger quantum systems a comprehensive and satisfactory description of  SLOCC-classes structure has not been established yet. 
Several studies in this area lead us in constructing invariants under LU and SLOCC transformations for three parties in \cite{threeQubits}, multipartite pure quantum systems in \cite{Wang_2014,RaG,Vrana2} the mixed multipartite states in \cite{Warna}.
It is known that that polynomial invariants are completely characterizing LU-equivalence classes \cite{PhysRevA.88.042304} and that the number of non-zero polynomials is a function of a local dimension $d$. 
SLOCC invariant polynomials were found for three qubit systems \cite{PhysRevLett.80.2245,PhysRevA.61.052306}; four qubit systems \cite{PhysRevLett.86.5188,PhysRevA.67.042303} and generally for multi-qubit systems \cite{PhysRevA.69.052304,PhysRevA.72.012337}. 
In qudit case, an interesting attempt of providing polynomial invariants has been also taken \cite{Szalay_2012}. 
Despite many attempts of classification of these polynomials \cite{SLOCCallDim}, and enhancing them with the physical meaning, their structure remains inscrutable. 
Many efforts have also focused on LU and SLOCC equivalence of stabilizer states \cite{gottesman1997stabilizer}, matrix-product states and projected entangled pair states \cite{8589272}, Gaussian states \cite{RevModPhys.77.513,PhysRevA.97.042325}, Locally Maximally Entangleable States (LMESs) \cite{PhysRevA.79.052304}, or Generalized Bell states \cite{PhysRevA.98.022304}. 

The initial motivation for our work was the question whether different constructions of AME states are equivalent by any local transformation. 
It was already shown that some $k$-uniform states are not LU (and in fact not SLOCC) equivalent \cite{raissi2019new}. 
This result was based on a comparison of the ranks of reduced density matrices \cite{RanksReducedMatrices}. 
Nevertheless, in some specific cases, the aforementioned rank argument is never conclusive. 
This is the case when two $k$-uniform states are of minimal support, or when the bound on $k$ is saturated, i.e. in the case of AME states. 
In this paper, we develop techniques of SLOCC-verification between such states. 
We provide methods of SLOCC-equivalence verification for all $k$-uniform states with minimal support.
We show that the conjecture of all AME states being LU- and SLOCC-equivalent does not hold. 
In particular, we show that some AME states cannot be transformed into existing minimal support form by any local invertible operation. 
Moreover, the class of LU/SLOCC-transformation is widely investigated and the systematic way for verification of LU/SLOCC-equivalence of AME states and $k$-uniform with minimal support is provided. 
We expose the vital contrast between AME states of small systems (up to 5 parties) and larger systems. 
In particular, for larger systems there exist infinitely many non-LU/SLOCC-equivalent AME states of minimal support differing only by phases. 
Additionally, we emphasise the essential difference between local transformation of $k$-uniform states or AME states of odd number of parties ($k<N/2$) and AME states of even number of parties ($k = N/2$). 
The structure of the latter is more complex and non-classical in some sense. 
Despite the refined analysis of this case, obtained results are still intricate and dependent on specific cases. 

The paper is organized as follows. 
In \cref{sec1}, we recall construction methods of $k$-uniform and AME states known from the literature, we provide several explicit examples of such states. Moreover, we discuss the relation between LU- and SLOCC-equivalences restricted to $k$-uniform and AME states. 
The main results of our paper are presented in the following three consecutive sections. 
\cref{31} discusses local transformation of $k$-uniform states or AME states of odd number of parties ($k<N/2$). 
Similar results concerning AME states of even number of parties ($k = N/2$) are presented in \cref{32}. 
\cref{33} applies general results obtained in the previous sections. 
We present several examples of LU and SLOCC non-equivalent $k$-uniform states. 
The precise number of LU- and SLOCC-classes of AME states with minimal support is specified. 
Moreover, the non-trivial bounds on the number of such classes of general AME states are given. 
In \cref{Combinatorial designs}, we discuss some classes of combinatorial designs directly related to our problem. 
Existence and extension of those designs turned out to be crucial to obtain the aforementioned results. 
This phenomenon is presented in detail. 
Summary and conclusions are presented in \cref{Conclusions}. 
Further discussion and open problems are left for \cref{Further discussion and open problems}. 
Proofs of statements included in \cref{31} might be found in \cref{app1,verification2APP,AppD}; whereas claims presented in \cref{32} are justified in \cref{The proof Case II}.

\section{Notation and Preliminaries}		
\label{sec1}
		
\subsection{$k$-uniform states and AME states}		

Consider a multipartite quantum state $\ket{\psi} \in \mathcal{H}_d^{\otimes N}$ of $N$ parties with a local dimension $d$ each. 
We say that $\ket{\psi}$ is a \emph{$k$-uniform} state if its reduced density matrices are maximally mixed, i.e.
\[
\rho_S (\psi ) \propto \Id
\]
for any subsystem $S$ of $k$ parties ($|S|$=k). 
The uniformity $k$ cannot exceed $\lfloor N/2 \rfloor$ \cite{PhysRevA.69.052330}. 
States which saturate this bound, i.e. $\lfloor N/2 \rfloor$-uniform states, are called \emph{Absolutely Maximally entangled} (AME) states, and are denoted by AME(N,d).
Particular attention is paid to AME states of an even number of parties, which are equivalent to notions as perfect tensors \cite{Pastawski2015HolographicQE} or multiunitary matrices \cite{MultiUnitary}.

The \textit{support} of a state $\ket{\psi}$ is the number of non-zero coefficients when $\ket{\psi}$  is written in the computational basis. Note, that the support of $k$-uniform state is at least $d^{N-k}$. Indeed, the partial trace over $N-K$ particles is an identity matrix $\Id_{d^{N-k}}$. $k$-uniform states with support equal to $d^{N-k}$ are called of \textit{minimal support}. 
$k$-uniform states are a natural generalization of well-established  GHZ state. 

\begin{example}
\label{ex1}
Greenberger–Horne–Zeilinger (GHZ) state
\[
\ket{\text{GHZ}} = \dfrac{1}{\sqrt{2}} \Big( \ket{00 0}  +\ket{111} \Big)
\] 
is a $1$-uniform state of minimal support. 
Similarly, its natural generalization to $N$ parties qudit states (each partie has exactly $d$ distinguishable energy levels):
\[
\ket{\text{GHZ}_d^N} =\dfrac{1}{\sqrt{d}}\Big( \ket{0\cdots 0} +\cdots +\ket{d-1 \cdots d-1} \Big)
\] 
is $1$-uniform state of minimal support. 
\end{example}

It is worth mentioning that GHZ state is maximising entanglement properties among all three qudit states. This statement, however, is not true anymore for larger systems \cite{SpeeKraus}.

\begin{example}
\label{ex2}
The following state of four qutrits
\begin{align*}
\ket{\text{AME(4,3)}} = \dfrac{1}{3} \Big(
&\ket{0000}+ \ket{0121}+ \ket{0212}+\\
&\ket{1110}+ \ket{1201}+ \ket{1022}+&&\\
&\ket{2220}+ \ket{2011}+ \ket{2102} \Big)&&
\end{align*}
is an AME(4,3) state of minimal support \cite{Helwig2013AbsolutelyME}. 
It reveals larger entanglement properties than a relevant GHZ state.
\end{example}

AME(N,d) states are maximizing entanglement properties among all $N$-parties states, each with $d$ levels \cite{HelwigAME}. 
There is no general construction of AME(N,d) state, and in fact, they do not exist for any numbers $N$ and $d$. 
Indeed, it was first observed that AME state of four qubits does not exist \cite{HIGUCHI2000213}. 
Nowadays, more of such negative results are known \cite{Felix72}. 
Some cases, as AME(4,6) are believed to not exists, despite the mathematical proof is still missing \cite{ComDesi}. 

We would like to finish this section with two remarkable observations. Firstly, all known $k$-uniform and AME states might be written by simple closed formulas. For instance, 
\begin{align}
\ket{\text{GHZ}} =&\dfrac{1}{\sqrt{d}}\sum_{i=0}^{d-1}\ket{i,\ldots, i},\nonumber\\ 
\ket{\text{AME(4,3)}} =&\dfrac{1}{d} \sum_{i,j=0}^{d-1}\ket{i,j,i+j,2i+j}, \label{AME43p} 
\end{align}
are relevant to GHZ state and AME(4,3) presented in \cref{ex1,ex2} 

Secondly, not all k-uniform states are of minimal support. 
It is rather easy to verify that AME states with minimal support of five or six qubits do not exist. 
Nevertheless, the construction of AME(6,2) was provided \cite{Rains1999NonbinaryQC,ComDesi}.  
We present one example of AME states with non-minimal support relevant to the future discussion.

\begin{example}
\label{ex3}
Consider the following states
\begin{align*}
\ket{\text{AME(5,d)}} =\dfrac{1}{\sqrt{d^3}} \sum_{i,j,k=0}^{d-1} &\omega^{(3i+j) k} \\
&\ket{i,j,i+j,2i+j +k,k} ,
\end{align*}
where $\omega$ is $d$th root of unity. 
$\ket{\text{AME(5,d)}} $ satisfy all properties required from AME states for any integer number $d\geq 2$ \cite{ComDesi,Rains1999NonbinaryQC}. 
They cannot be written, however, in the minimal support form.
\end{example}


\subsection{Orthogonal arrays}

Orthogonal arrays \cite{OA} are combinatorial arrangements, tables with entries satisfying given orthogonal properties. They were created in response to optimization problems in statistical analysis.. Their most famous application can be summarized in one sentence: "Your automobile lasts longer today because of orthogonal arrays" \cite{OAcars}.

A close connection between OAs and maximally entangled states \cite{DiK}, error-correcting codes \cite{OA} brought a new life for these combinatorial objects.  
Forasmuch some of OAs might are in one-to-one correspondence with $k$-uniform states, the concept of OA is briefly presented below.

An orthogonal array $\oa{r,N,d,k}$ is a table composed by $r$ rows, $N$ columns with entries taken from $0,\ldots,d-1$ in such a way that each subset of $k$ columns contains all possible combinations of symbols with the same amount of repetitions. The number of such repetitions is called \textit{the index} of the OA and denoted by $\lambda$. One may observe, that the index of OA is related to the other parameters:
\[
\lambda =\dfrac{r}{d^k} .
\]
OA with $\lambda =1$ is called OA of \emph{Index Unity}. 
\cref{OA1} presents an example of an Index Unity OA. 

A pure quantum state consisting of $r$ terms might be associated with $\oa{r,N,d,k}$, simply by reading all rows of OA \cite{DiK,OA}. 
With a little more effort, one may adjust phases $\omega_1,\cdots , \omega_r$ in front of any term, see \cref{eq2}. Intriguingly, this relevance provides a one-to-one correspondence between $k$-uniform states of the minimal support and OAs of index unity.

\begin{proposition}
There is one-to-one correspondence between $k$-uniform states with the minimal support of $N$ qudits and OA $\oa{d^k,N,d,k}$ enhanced with the phase vectors
\[
\big( \omega_1 ,\ldots ,\omega_{d^k} \big) ,
\]
where $|\omega_i |=1$.
\end{proposition}

Form $\oa{d^k,N,d,k}$ the $k$-uniform state is created by reading all terms and adjusting them with the relevant phases $\omega_i$.
Conversely, from AME state the OA of index unity might be built, simply by erasing phases and adjusting all terms one above the other. 

\begin{example}
\label{eq2}
For any phases $|\omega_{i,j} | =1 $ the following state is 2-uniform:
\begin{equation*}
\ket{\text{AME(4,3)}_\omega} =\dfrac{1}{d} \sum_{i,j=0}^{d-1} \omega_{i,j} \ket{i,j,i+j,2i+j}.
\end{equation*}
It is defined uniquely up to global phase.
\end{example}

\begin{figure}[h!]
\centering
\[
 \begin{array}{*4{c}}
    \tikzmark{left}{0} &
0&0&0 \\
0&1&1&1\\
0&2&2&2\\
1&0&1&2\\
1&1&0&2\\
1&2&0&1\\
2&0&2&1\\
2&1&0&2\\
2&    \tikzmark{right}{2}&1&0
\Highlight[first]
  \end{array}
  \qquad
  \begin{array}{*5{c}}
0&
\tikzmarkk{up}{0}&0&0 \\
0&1&1&1\\
0&2&2&2\\
1&0&1&2\\
1&1&0&2\\
1&2&0&1\\
2&0&2&1\\
2&1&0&2\\
2&  2&\tikzmarkk{down}{1}&0
\Highlightt[first]
  \end{array}
      \qquad
 \begin{array}{*4{c}}
\ket{\text{AME(4,3)}}&\propto & \omega_{00}&\ket{0000} \\
 &+&\omega_{01}&\ket{0111} \\
&+&\omega_{02}&\ket{0222}\\
&+&\omega_{10}&\ket{1012} \\
&+&\omega_{11}&\ket{1120}\\
&+&\omega_{12}&\ket{1201}\\
&+&\omega_{20}&\ket{2021} \\
&+&\omega_{21}&\ket{2102} \\
&+&\omega_{22}&\ket{2210}\\
  \end{array}
\]
\Highlightt[second]
\caption{The orthogonal array of unity index $\oa{9,4,3,2}$ on the left and repeated in the center. Each subset consisting of two columns contains all possible combinations of symbols. Here, two such subsets are highlighted. The relevant quantum state is obtained by forming a superposition of states corresponding to consecutive rows of the array enhanced by some phases -- see the expression
on the right.}
\label{OA1}
\centering
\end{figure}

It is known that for any number $N$ and $k$ there exists $\oa{d^k,N,d,k}$ for a local dimension $d$ being sufficiently large (in fact, such construction is given for $d$ being a prime power satisfying $d>k$ and $d>N-1$) \cite{Bush}. Hence, for any number of parties $N$ the $k$-uniform state with minimal support might be created where the component systems have a sufficiently large number of levels. 
The problem of existence and classification of OAs (in particular OAs of index unity) has been extensively studied \cite{HedayatIndex1OA,BushStudies}. 
We refer to the web page of N. J. A. Sloane for tables of OAs \cite{OAlib}.

\begin{example}
\label{ex5}
The following states
\begin{align*}
\ket{\text{AME(5,d)'}} =&\dfrac{1}{d} \sum_{i,j=0}^{d-1} \ket{i,j,i+j,2i+j,3i+j} 
\end{align*}
are AME state with minimal support for all $d\geq 5$ being prime numbers \cite{AME-QECC-Zahra}. 

In fact, with a little more effort such states might be constructed for all prime powers $d \geq 4$ \cite{AME-QECC-Zahra}.
For instance, the following state
\begin{align*}
\ket{\text{AME(4,4)}} =&\dfrac{1}{4} \sum_{i,j=0}^{d-1} \ket{i,j,M^1_{i,j},M^2_{i,j}}
\end{align*}
where
 \begin{align*}
 \mbox{\small$\displaystyle
  M^1 :=
 \begin{pmatrix}
0& 1& 2& 3\\
1& 0 &  3& 2\\
2& 3 & 0 &1\\
3& 2& 1& 0
\end{pmatrix}
, \quad
 M^2 :=
 \begin{pmatrix}
0& 2& 3& 1\\
1& 3 &  2& 0\\
2& 0 & 1 &3\\
3& 1& 0& 2
\end{pmatrix}
;
$}
\end{align*}
is an AME(4,4) state of minimal support (rows and columns of $M^1, M^2$ are indexed by $i,j=0,\ldots,3$) \cite{AME-QECC-Zahra}. Construction of matrices $M^1$ and $M^2$ comes from the multiplication structure in the Galois field GF(4), which might be seen as a multiplication of irreducible polynomials of degree 2 \cite{AME-QECC-Zahra}. 
In fact, matrices $M^1, M^2$ form a mutually orthogonal Latin square MOLS(4), see \cref{MOLH} for details.
\end{example}

As we already mentioned in \cref{ex3}, not all k-uniform states are of minimal support, which simply means that not all AME states are obtained from OAs of index unity. 

\subsection{Composed systems}

For any two $k$-uniform states $\ket{\psi_1}$ and $\ket{\psi_2}$, one may consider the composed system $\ket{\psi_1} \otimes \ket{\psi_2}$, which 
inherits the property of being $k$-uniform. For instance, the following state:
\begin{align}
\label{composedAME}
\ket{\text{AME(4,9)}_{3\times 3}} =\dfrac{1}{9} 
\sum_{\substack{i,j=0\\ k, \ell=0}}^{2} &
\big\vert (i,k),(j,\ell),\\ 
&(i+j,k+\ell),(i+2j,k+2\ell) \big\rangle \nonumber
\end{align}
is a composition of two $\ket{\text{AME(4,3)}} $ states from \cref{AME43p}. Each pair $(i,k)$ is identified with a number $0,\ldots,8$ written in the ternary numeral system, i.e. $(i,k )\cong 3i+j $.
\subsection{Local transformations}

Two $N$ qudit states $\Ket{\psi}$ and $\Ket{\phi}$ are LU-equivalent if one can be transformed into another by local unitary operators, i.e
\[
\Ket{\phi} = U_1 \otimes \cdots \otimes U_N \Ket{\psi}.
\]
The LU-equivalence of $\Ket{\psi}$ oneself is referred in the text as an \emph{automorphism}. 

Mathematically, two states $\Ket{\psi}$ and $\Ket{\phi}$ are SLOCC-equivalent if and only if there exists a \emph{local invertible} operator connecting those states \cite{ThreeQub}:
\[
\Ket{\phi} = O_1 \otimes \cdots \otimes O_N \Ket{\psi}.
\]
Since LU- and SLOCC-equivalences are equivalence relations, the state space might be naturally partitioned into \emph{LU classes} and \emph{SLOCC classes} respectively.

\subsection{The structure of SLOCC classes}

We present a brief outline of some algebraic invariant methods interconnected to the SLOCC partition problem of multipartite entangled states. 
We introduce the notion of critical states and we discuss consequences of Kempf-Ness theorem \cite{KempfNess} for multipartite systems. 
We refer to \cite{KempfNessToEntanglement} for more details. 

The state $\rho$ is called a \emph{critical state} if its all reduced density matrices $\rho_i$ are proportional to the identity. 
In particular, the class of critical states contains stabilizer states, cluster states, and all $k$-uniform states among many others \cite{KempfNessToEntanglement}.

Notice, that the critical states were initially defined differently, via an action of the Lie group associated with the state space. 
By applying Kempf-Ness theorem \cite{KempfNess} it was later observed that, indeed, states are critical if and only if they are maximally entangled \cite{GourWallach}. 

Kempf-Ness theorem has one more significant consequence for multipartite quantum states. It follows that within one SLOCC class, the critical states are unique up to LU-equivalences. Therefore, such classes posses the canonical representative. 

Notice that not all SLOCC classes contain a critical state, and hence there is no one-to-one correspondence between SLOCC classes and maximally mixed states. More precisely, each SLOCC class is topologically closed (equivalently closed with respect to Zariski topology) if and only if it contains a critical state \cite{sowik2019link}. In fact, closed SLOCC classes are dense in a state space \cite{NotesWallach}.

We conclude this discussion with the following corollary. 

\begin{corollary}
\label{LU=SLOCC}
Two critical states are in the same SLOCC class if and only if they are LU equivalent. 
Notice, that all $k$-uniform states are critical states.
\end{corollary}

Therefore, verification of LU-equivalence between two $k$-uniform states is equivalent to verification of SLOCC-equivalence between them.


\section{Local equivalences, case $2k<N$}
\label{31}
We introduce another class of local unitary operations, essential for the classification problem of $k$-uniform states.

\begin{definition}
A unitary matrix $M$  is called \textit{unitary monomial matrix} if one of the following holds:
\begin{enumerate}
\item $M$ has exactly one nonzero entry in each row and each column;
\item $M$ is a product of a permutation and diagonal matrix;
\item $M$ does not change the support of any quantum state
\end{enumerate}
For multiparty systems, the local monomial operation will be denoted as $LM$-equivalency.
\end{definition}

One can see that all conditions: 1-3 are, indeed, equivalent. 
Obviously, each local monomial operation provides the LU-equivalence between two states of minimal support. 
Indeed, since it is a local unitary operation, it does not change the entanglement properties of a state; moreover, it preserves the number of elements in the support of a state. 
As we shall see, the reverse statement is also true for $k$-uniform states where $2k<N$. In other words, searching for LU-equivalence between such two $k$-uniform states of minimal support might be restricted to the search within the $LM$ class.

\begin{proposition}
\label{prop1}
For $2k<N$, each LU- or SLOCC-equivalency between two $k$-uniform states of minimal support is in fact LM-equivalency.
\end{proposition}

\begin{corollary}
\label{coro1}
For $2k<N$, two $k$-uniform states of minimal support are LU- or SLOCC-equivalent if and only if they are LM-equivalent.
\end{corollary} 

We shall prove the above statement in a slight enhanced version in \cref{app1}. 


With the strengthened version of \cref{prop1} at hand (see \cref{app1}), we have shown that not all AME states are equivalent, which was an initial motivation for our research.


\begin{proposition}
\label{AME55}
Two families of AME(5,d) states: $\ket{\text{AME(5,d)}}$ and $\ket{\text{AME(5,d)'}}$
presented in \cref{ex3} and \cref{ex5} respectively, are not LU-equivalent for any prime local dimension $d$. 
\end{proposition}

Therefore, as an immediate conclusion form \cref{LU=SLOCC}, states $\ket{\text{AME(5,d)}}$ and $\ket{\text{AME(5,d)'}}$ belongs to different SLOCC classes. 
We refer to \cref{AppD} for the proof of the above statement.

Both states $\ket{\text{AME(5,d)}}$ and $\ket{\text{AME(5,d)'}}$ belongs to special classes of quantum states: \emph{stabilizer} states and \emph{graph} states. 
A comprehensive introduction to this topic might be found in \cite{PhysRevA.84.052306}. 
The class of natural local operation among stabilizer/graph states is called \emph{Local Clifford operations} (LC). 
The verification of  LC-equivalence between two stabilizer states is rather a simple problem and might be the polynomial-time algorithm \cite{Clifford}. 
Nevertheless, LU-or SLOCC-equivalence verification of such states is generally a challenging problem. 
Surprisingly, it was shown that there are LU-equivalences of graph/stabilizer states which are beyond LC-class \cite{LUnoLC}. 
One can see our results as proof for not LU-equivalence of two graph/stabilizer states. 

\subsection{Verification of LU-equivalence}
\label{IIIA}

The problem of verification whether two different states are LU/SLOCC-equivalent is of the most importance from the application point of view \cite{SLOCCver}. 
Monomial matrices are products of permutations and diagonal matrices. Although permutation matrices are easy to quantify, diagonal matrices are indexed by a real coefficient. 
The following statement shows that we can overcome this apparent difficulty. 
In fact, verification of LU-equivalence between two states is restricted to testing a finite number of possible equivalences.

We introduce the following notation. 
Consider a $k$-uniform state with minimal support 
\[
\ket{\psi} =\sum_{I \in \mathcal{I}} \omega_I \ket{I} ,
\] 
where $I \in [d]^{n}$ are multi-indices running over the set $\mathcal{I}$ of the size $| \mathcal{I}| =d^k$. 
Denote by ${\mathcal{I}^i_a}$ all those indices with $a$ on the $i$-th position and by
\[
W^i_a := \prod_{I \in \mathcal{I}^i_a} \omega_{I} .
\] 
Similarly, ${\mathcal{I}^{i_1 ,\ldots i_\ell}_{a_1 ,\ldots a_\ell}}$ denotes the set of indices with $a_i$ on the $i$-th position,
\[
W^{i_1 ,\ldots i_\ell}_{a_1 ,\ldots a_\ell} := 
\prod_{I \in \mathcal{I}^{i_1 ,\ldots , i_\ell}_{a_1 ,\ldots , a_\ell}} \omega_{I} .
\] 
Observe, that any local permutation $\sigma = \sigma_1 \otimes\cdots\otimes\sigma_n$ act on $W^{i_1 ,\ldots i_\ell}_{a_1 ,\ldots a_\ell} $ by permuting relevant elements:
\[
\sigma \Big( W^{i_1 ,\ldots i_\ell}_{a_1 ,\ldots a_\ell} \Big) :=
\prod_{I \in \mathcal{I}^{i_1 ,\ldots , i_\ell}_{\sigma_{i_1} (a_1 ) ,\ldots , \sigma_{i_\ell} (a_\ell )}} \omega_{I} .
\]

We specify the diagonal matrices which might appear in local equivalences described in \cref{prop1}.

\begin{proposition}
\label{verification}
Consider two $k$-uniform states of minimal support $\ket{\psi}$ and $\ket{\psi '}$. The eventual LU-equivalence between them is of the following form:
\[
\omega \Big( \sigma_1 D_1 \otimes \cdots \otimes \sigma_n D_n  \Big)
\]
where $\omega$ is a global phase, $\sigma $ is a local permutation $\sigma =\sigma_1, \ldots ,\sigma_n$, and $D_i$ are a following diagonal matrices:
\[
D_i =\text{diag} \Bigg( \sqrt[\leftroot{-3}\uproot{3}d]{\tfrac{(W^{i,S}_{0,I})'}{\sigma (W^{i,S}_{0,I})}},
\ldots ,
\sqrt[\leftroot{-3}\uproot{3}d]{\tfrac{(W^{i,S}_{d-1,I})'}{\sigma (W^{i,S}_{d-1, I})}} \Bigg),
\]
with entries given by any $d$th root of a relevant complex number, 
where $S$ is any subset of $k-2$ indices which do not contain $i$, and $I$ is any multi-index $I=i_2,\ldots ,i_{k-1}$. 
In particular, for $k=2$ it is the empty set $S\equiv \emptyset$.

Moreover, for $k>2$ there is a following necessarily condition for existence of such LU-equivalence. 
For any $S \subset [n] \setminus \{i \}$, such that $|S| = k-2$ and any symbol $\ell$ 
\[
\tfrac{ (W^{i,S}_{\ell,I})'}{\sigma (W^{i,S}_{\ell,I})} =\tfrac{(W^{i,S}_{\ell,I'})'}{\sigma (W^{i,S}_{\ell,I'})}
\]
for arbitrary multi-indices $I=i_2,\ldots ,i_{k-1}$ and $I'=i_2',\ldots ,i_{k-1}'$.
\end{proposition}

The proof of \cref{verification} is given in \cref{verification2APP}. 
We illustrate the usefulness of this criterion on two examples: $2$- and $3$-uniform states, see \cref{33}.
 

\section{Local equivalences, case $2k=N$}
\label{32}

For AME(2k,d) states, the statement of \cref{prop1} does not hold anymore. 
The following example illustrates this difference. 

\begin{example}
The Fourier transform $F_3$:
\begin{equation}
\label{F(3)}
\big(F_3\big)^{\otimes 4}  =
\begin{pmatrix} 
1 & 1 & 1 \\ 
1 & \omega & \overline{\omega} \\ 
1& \overline{\omega} & \omega
\end{pmatrix}^{\otimes 4} 
\end{equation}
provides an automorphism of AME(4,3) state from \cref{AME43p}. 
\end{example}

As we shall see, 
Fourier matrices are not the only non-monomial matrices providing the LU-equivalence between AME(2k,d) states of minimal support. 
On the other hand, it is not generally true that Fourier matrices preserve all AME(2k,d) states with minimal support. 
Despite the exhaustive analysis performed, the general structure of LU-equivalences between AME(2k,d) states is still puzzling and remains unknown. 
Nevertheless, for sufficiently small values of $k$ and $d$, the complexity of the problem reduces significantly; 
and the analogue of \cref{prop1} might be stated. 
We begin with the definition of matrices beyond monomial class, which can provide LU-equivalence between AME states.

\begin{definition}
Let $d$, $q$ be positive integers. A \emph{Butson-type} complex Hadamard matrix of order $d$ and complexity $q$ is a unitary matrix in which each entry is a complex $q$th root of unity scaled by the factor $1/d$. 
The set of Butson-type matrices is denoted by $\text{BH}(d,q)$.
\end{definition}

In literature \cite{Butson,Tadej}, Butson-type matrices are defined without scaling factor $1/d$, therefore they are proportional to the unitary matrices. 

For our purpose, it is enough to focus on matrices of the type $\text{BH}(d,d)$. As we shall see, those matrices (up to monomial matrices) defines all possible LU-equivalences of AME states. 
The problem of existence and classification of such matrices is discussed later, in \cref{Butson}.

\begin{proposition}
\label{prop1=}
Consider two AME($2k$,$d$) states of minimal support, where $k$ and $d$ are sufficiently small (see \cref{small}). 
Each LU-equivalency between them is of one of the following forms:
\begin{enumerate}
\item tensor product of Butson-type matrices $ B_i \in \text{BH}(d,d)$ multiplied by LM matrices from each side; 
\item or LM matrices itself.
\end{enumerate}
\end{proposition}

Similarly to the case $2k<n$, we can specify the class of diagonal matrices which might appear in LM-equivalences from \cref{prop1=}.

\begin{proposition}
\label{verification2}
LU-equivalence between two AME($2k$,$d$) states of minimal support, where $k$ and $d$ are sufficiently small (see \cref{small}) is of the following form: either
\begin{equation}
\label{AMEeq1}
\omega \Big( \overrightarrow{D}_1^{-1} B_1   \overleftarrow{D}_1^{-1}
 \otimes \cdots \otimes
 \overrightarrow{D}_n^{-1}   B_n   \overleftarrow{D}_n^{-1}  \Big),
\end{equation}
or
\begin{equation}
\omega \Big( \sigma_1 D_1 \otimes \cdots \otimes \sigma_n D_n  \Big)
\label{AMEeq2}
\end{equation}
where $\omega$ is a global phase, 
$B_i \in \text{BH}(d,d)$ are Butson-type matrices, 
$ \sigma_i$ are permutation matrices, 
and $  D ,\overrightarrow{D}_i ,\overleftarrow{D}_i$  are diagonal matrices.

The entries of diagonal matrices are $d$th root of specified complex numbers:
\begin{align*}
D_i =
&
\text{diag} \Bigg( \sqrt[\leftroot{-3}\uproot{3}d]{\tfrac{(W^{i,S}_{0,I})'}{\sigma_i (W^{i,S}_{0,I})}},
\ldots ,
\sqrt[\leftroot{-3}\uproot{3}d]{\tfrac{(W^{i,S}_{d-1,I})'}{\sigma_i (W^{i,S}_{d-1, I})}} \Bigg), 
\\
\overrightarrow{D}_i =
&
\text{diag} \Bigg( \sqrt[\leftroot{-3}\uproot{3}d]{(W^{i,S}_{0,I})'},
\ldots ,
\sqrt[\leftroot{-3}\uproot{3}d]{(W^{i,S}_{d-1,I})'} \Bigg), 
\\
\overleftarrow{D}_i=
&
\text{diag} \Bigg( \sqrt[\leftroot{-3}\uproot{3}d]{W^{i,S}_{0,I}},
\ldots ,
\sqrt[\leftroot{-3}\uproot{3}d]{W^{i,S}_{d-1,I}} \Bigg), 
\end{align*}
where $S$ is any subset of $k-2$ indices which do not contain $i$, and $I$ is any multi-index $I=i_2,\ldots ,i_{k-1}$, and $\sigma $ is a local permutation $\sigma =\sigma_1, \ldots ,\sigma_n$ of levels. 
In particular, for $k=2$ it is the empty set $S\equiv \emptyset$.

Moreover, for $k>2$ there are a following necessarily condition for existence of such equivalence. 
Consider any $S \subset [n] \setminus \{i \}$, such that $|S| = k-2$ and any symbols $\ell ,\ell'$. 
For arbitrary multi-indices $I=i_2,\ldots ,i_{k-1}$ and $I'=i_2',\ldots ,i_{k-1}'$:
\[
\tfrac{ (W^{i,S}_{\ell,I})'}{\sigma (W^{i,S}_{\ell,I})} =\tfrac{(W^{i,S}_{\ell,I'})'}{\sigma (W^{i,S}_{\ell,I'})}
\]
if the equivalence is of the form \cref{AMEeq1}, and 
\begin{align*}
\tfrac{ W^{i,S}_{\ell,I} }{W^{i,S}_{\ell',I} } &
=\tfrac{ W^{i,S}_{\ell,I'} }{W^{i,S}_{\ell',I'} }  \\
\tfrac{( W^{i,S}_{\ell,I} )'}{(W^{i,S}_{\ell',I} )'} &
=\tfrac{ (W^{i,S}_{\ell,I'})' }{(W^{i,S}_{\ell',I'})' }  
\end{align*}
if the equivalence is of the form \cref{AMEeq2}.
\end{proposition}

The statement above is important from the application point of view. 
The enormous class of diagonal matrices is significantly restricted. 
In fact, the classes of matrices from \cref{prop1=} which may provide the LU-equivalence between two states is brought to be finite. 
Therefore, the LU-equivalence verification problem is discretized and made finite. 
Moreover, the second part of \cref{verification2} impose some necessary conditions for two AME(2k,d) states to be LU-equivalent for $k>2$. 
As we shall see, such assumptions might be easily validated, which implies the existence of not LU/SLOCC-equivalent AME(2k,d) states for all $k>2$, see \cref{3-uniform states}. 
For the proof of \cref{prop1=,verification2}, we refer to \cref{The proof Case II}.

\begin{remark}
\label{small}
There is the following restriction on numbers $d$ and $k$ imposed in the statement of \cref{prop1=} and \cref{verification2}:
\[
\Big(k+1 \Big) \Big( 1 + \sqrt[\leftroot{-3}\uproot{3}{k-1}]{k} \Big) \leq d .
\]
for $k>1$ and $2<d$ for $k=1$. 
This bound is related to the necessary condition for existence and extension of combinatorial designs called \emph{mutually orthogonal hypercubes}. 
We discuss them in detail in \cref{Combinatorial designs}. 

In particular, for $k=2,3,4,5,6$ the smallest value of $d_{\text{min}}$ which does not satisfy the bound above is presented in a table. 
The given bound is not tight. 
Moreover, we present $d_{\text{max}}$  the maximal value of a local dimension $d$ for which \cref{prop1=} does not holds (we found a counterexample). 
In particular, for the local dimension $d_{\text{max}}$ we found LU-equivalences which is not of the form presented in \cref{prop1=}. 
The origin of $d_{\text{max}}$ is presented in \cref{ComposedSyst}.

\begin{table}[h!]
  \begin{center}
    \begin{tabular}{|c|c|c|} 
    \hline
      AME(2k,d) & $d_{\text{min}}$ &$d_{\text{max}}$ \\
      \hline
      AME(2,d) & $ 3$& $ 3$\\
      AME(4,d) & $ 9$& $ 9$\\
        AME(6,d) & $ 11$& $ 16$\\
        AME(8,d) & $ 13$& $ 25$\\
        \hline
    \end{tabular}
  \end{center}
\end{table}
\noindent
We conjecture, that those values behave asymptotically as $(k-1)^2$.

In fact, assumptions on values $k$ and $d$ are not restrictive from the application point of view. 
Indeed states outside the described class are far beyond current laboratory possibilities \cite{10.1038/35005011,10.1038/s41567-019-0508-6}.
\end{remark}

\subsection{Composed systems}
\label{ComposedSyst}

Consider AME(4,9) state being a product of two AME(4,3) states as it was described in \cref{composedAME}. 
Since the Fourier transform $(F_3)^{\otimes 4}$ and the identity $\Id_3^{\otimes 4}$ are automorphisms of AME(4,3), obviously
\begin{equation}
\label{prodForm}
\Big( F_3 \otimes \Id \Big)^{\otimes 4}
\end{equation}
provides an automorphism of aforementioned AME(9,3) state. 
One immediately observes that \cref{prodForm} it is not of the form postulated in \cref{prop1=}. 
Indeed, according to \cref{prop1=} matrices providing LU-equivalence between two AME states of minimal support have either $1$ or $d$ non-zero entries in each row and column; each entry of the same norm. Matrices from \cref{prodForm}, however, have exactly $3$ non-zero entries in each row and column. 
Nevertheless, $k$ and $d$ were assumed to be sufficiently small in \cref{prop1=};  
indeed, according to \cref{small}, the statement was restricted to $d<9$ for AME(4,d) states. 

Notice that the similar automorphisms might be potentially given for any product states. 
We conjecture, that LU-equivalences of such a product form are the only ones violating the statement of \cref{prop1=}.

\begin{conjecture}
\label{conProd}
If the LU-equivalence $U_1 \otimes \cdots \otimes U_{2k}$ of two AME states of minimal support is not of the form proposed in \cref{prop1=} for $k>1$; those states are product states and $U_i$ splits according to the composition of states into matrices postulated in \cref{prop1=}.
\end{conjecture}

The conjecture excludes AME(2,d) where the structure of LU-equivalences is more abundant, see \cref{1-uniform states}.

Even though \cref{conProd} seems reliable, the mathematical proof of it is out of reach at this stage of a research. 
In the most general case (without any assumptions on $k $ and $d$, we showed that the following might be noted about LU-equivalence between AME states.

\begin{proposition}
\label{cor2}
Consider two AME states $\ket{\psi}$ and $\ket{\psi'}$ of minimal support which are locally equivalent by $U:= U_1 \otimes\cdots\otimes U_n$. 
\begin{enumerate}
\item 
Each row/column of each matrix $U_i$ has the same number $s$ of non-zero elements, all having the same norm $\sqrt{s}$.
\item 
The number $s$ satisfies extension property, namely, there exists MOLH($s$) which can be further extent onto MOLH($d$).
\item Under the assumption that all phases are trivial, i.e. $\omega_I \equiv \omega_I' \equiv 1$ for all multi-indices $I$, all non-zero entries of matrices $U_i$ are $s$th roots of unity (scaled by $\sqrt{s}$) up to global multiplication by a complex number.
\end{enumerate}
for the notion of mutually orthogonal Latin hypercubes (MOLH) we refer to \cref{Combinatorial designs}.
\end{proposition}

\noindent
For the proof, we refer to \cref{The proof Case II}.

\subsection{Classification of Butson-type matrices}
\label{Butson}

\cref{F(3)} provides a first unitary equivalence of AME states with minimal support which is not a monomial matrix. 
We have shown, that under some restrictions on $k$ and $d$ (see \cref{small}), such equivalences are local Butson matrices $ B_i \in \text{BH}(d,d)$ up to local monomial transformations. 
It is not difficult to show that Fourier transform and tensor products of such are elements of the class $ \text{BH}(d,d)$. Nevertheless the class of Butson matrices is much larger and contains $1,2,1,4,1,143,23,51,1,4497733$ matrices (classified up to monomial matrices) for $d=3, \ldots,12$ \cite{Lampio}. Tables of Butson matrices are available in \cite{BUlib,BUlib2}. 

Even though, the class of Butson matrices $ \text{BH}(d,d)$ grows rapidly with $d$; it seems that the subclass of such matrices that might be involved in LU-equivalences of AME states is significantly smaller. 
Therefore, more specific classification of matrices providing eventual LU-equivalences of AME states is needed. 
We suspect that all such matrices are Fourier transform and tensor product of such.

On the other hand, not all Fourier matrices might be involved in LU-equivalences of AME states. 
For $d>3$, we could not construct LU-equivalence of AME(5,$d$) states from \cref{ex5} based on Fourier matrices $F_5$. 
It is worth mentioning, that the Fourier transforms $F_5, F_7, F_{11}$ are the only matrices of type $ \text{BH}(p,p)$. for $p=5,7,11$.  
Hence, most probably all automorphisms of AME($p$,$p$) states from \cref{ex5} for $p=5,7,11$ are within LM class. 
Nevertheless, for $d=4$ such an equivalence might be provided by tensor product $F_2 \otimes F_2$, which is illustrated below.

\begin{example}
\label{ex9}
The action of 
\[
\Big(F_2 \otimes F_2 \Big)^{\otimes 4}
\]
on the state $\ket{\text{AME(4,4)}} $ from \cref{ex5} is equivalent to the following local permutation of indices:
\[
\mbox{\small$\displaystyle
 \begin{pmatrix}
1& \cdot & \cdot & \cdot \\
\cdot & \cdot & \cdot & 1 \\
\cdot & 1& \cdot & \cdot \\
\cdot & \cdot &1 & \cdot 
\end{pmatrix}
\otimes
 \begin{pmatrix}
1 & \cdot & \cdot & \cdot \\
\cdot & \cdot & 1 & \cdot \\
\cdot & \cdot & \cdot & 1 \\
\cdot & 1 & \cdot & \cdot 
\end{pmatrix}
\otimes
\Id
\otimes
 \begin{pmatrix}
1 & \cdot & \cdot & \cdot \\
\cdot & \cdot & \cdot & 1\\
\cdot & 1 & \cdot & \cdot \\
\cdot & \cdot & 1 & \cdot 
\end{pmatrix}
,
$}
\]
\noindent
and hence provides the LU-equivalence between AME states of minimal support.
\end{example}

We shall finish this section by linking the problem of decreasing the class of Butson matrices involved in LU-equivalences of AME states with two mathematical problems. 

Firstly, it has been conjectured that for prime dimensions $d$ the Fourier matrices $F_d$ are the only matrices in $ \text{BH}(d,d)$ \cite{Lampio}. 

Secondly, for some numbers $n_1$ and $n_2$ the tensor products of two Fourier matrices $F_{n_1}$ and $F_{n_2}$ is isomorphic to $F_{n_1 n_2}$. 
For instance, $F_2 \otimes F_3 \cong F_6$, while $F_2 \otimes F_2\neq \cong F_4$. The problem of determining those numbers has been solved \cite{Tadej2}.


\section{Existence and uniqueness of $k$-uniform states}
\label{33}

We apply so-far obtained results for various classes of AME states here. 
We present $1$-, $2$- and $3$-uniform states classes separately since the analysis of their LU/SLOCC-equivalences differs greatly. 

\subsection{1-uniform states}
\label{1-uniform states}

All $1$-uniform states of the minimal support are of the following form:
\[
\ket{\psi} =\dfrac{1}{\sqrt{d}} \sum_{i=0}^{d-1} \omega_i \ket{j^1_i ,\ldots ,j^N_i} ,
\]
where $j^\ell_i$ runs over all levels $0,\ldots,d-1$ for all indices $\ell$. 
One can observe that they are equivalent to the generalized Bell state $\ket{\text{GHZ}_d^N} $, see \cref{ex1}, and hence pairwise LU-equivalent. 
Indeed, the following local transformation 
\[
U_1 \big( \ket{j^1_i} \big) =  \big( \omega_i^{-1} \ket{j_i^1} \big) , \quad
U_\ell \big( \ket{j_i^\ell} \big) =  \big(  \ket{j^1_i} \big) 
\]
for systems $\ell=2,\ldots,N$ provides aforementioned LU-equivalence. 

\begin{observation}
all 1-uniform states of minimal support are LU-equivalent.
\end{observation}

This straightforward observation suggest that the structure of $1$-uniform states is rather simple and, in fact, not interesting. 
However, we point out some intriguing property is the automorphism group of AME(2,d) states. 
One can see that the Fourier transform $F_2$ preserves the Bell state:
\begin{equation*}
\begin{pmatrix} 
1 & 1  \\ 
1 & -1
\end{pmatrix}^{\otimes 2} \dfrac{1}{\sqrt{2}}\Big( \ket{00}+\ket{11} \Big) =
\dfrac{1}{\sqrt{2}}\Big( \ket{00}+\ket{11} \Big) .
\end{equation*}
Similarly, the Fourier transform: $F_n \otimes \overline{F_n}$ preserves the generalized Bell state of two parties:
\begin{equation}
\label{AME(2d)}
\ket{\text{AME(2,d)} }:= \dfrac{1}{\sqrt{d}} \Big( \ket{00}+ \cdots +\ket{(d-1)(d-1)} \Big) .
\end{equation}
Interestingly, all tensor product $U\otimes \overline{U}$ of unitary matrices preserves the generalized Bell state AME(2,d). 
Indeed, for each $i$:
\[
 U\otimes \overline{U} \ket{ii} =
 \sum_{j=0}^{d-1} 
 |u_{ij}|^2 \ket{jj} + \text{others} ,
\]
and hence
\begin{align*}
 U\otimes \overline{U} & 
 \ket{\text{AME(2,d)}} = \dfrac{1}{\sqrt{d}}\sum_{i=0}^{d-1} \sum_{j=0}^{d-1} 
 |u_{ij}|^2 \ket{jj} + \text{others} \\
&=\dfrac{1}{\sqrt{d}} \sum_{j=0}^{d-1} \Bigg( \sum_{i=0}^{d-1} 
 |u_{ij}|^2 \Bigg) \ket{jj} + \text{others} \\
 &=\dfrac{1}{\sqrt{d}} \sum_{j=0}^{d-1} \ket{jj} + \text{others} .
\end{align*} 
Since the state was normalized, all other terms on the right side disappear. 

This is in a contrast to AME(2k,d) states for $k>1$, where LU-equivalences were provided only by appropriate Butson type matrices $\text{B(} d,d \text{)}$ for all $d$ sufficiently small.

\subsection{2-uniform states}

Consider $2$-uniform states with minimal support:
\begin{align*}
\ket{\phi_\alpha} :=  \dfrac{1}{d} \Bigg(&  \alpha\ket{0,\ldots,0} + \\
&\sum_{i,j \neq (0,0)} \ket{i,j} \otimes \ket{\phi_{i,j}} \Bigg)
\end{align*}
indexed by a complex numbers $\alpha$, $|\alpha |=1$. 
Each of such a state is LU-equivalent to $\ket{\phi_{\alpha =0}} $ by the following:
\begin{align*}
U_1 =& \text{diag} \Big( ({\overline{\omega_\alpha}})^{n-1}, 
1 ,\ldots ,
1  \Big), \\
U_i= & \text{diag} \Big(
({\overline{\omega_\alpha}})^{d-1},
\omega_\alpha ,\ldots ,
\omega_\alpha \Big) ,
\end{align*}
for $i=2,\ldots,n$, where $\omega_\alpha = \sqrt[\leftroot{-3}\uproot{3}{d (n-1)}]{\alpha} $ is an arbitrary root.

Since all states from the family $ \ket{\phi_\alpha} $ are LU-equivalent with $\ket{\phi_{\alpha =0}}$, they are also pairwise equivalent. 
If the exceptional phase stands by a different term, the similar transformation of such a state onto $\ket{\phi_{\alpha =0}}$ might be given. 
Therefore, all $2$-uniform states: 
\begin{equation*}
\ket{\phi_{\omega}} =  \dfrac{1}{d} \sum_{i,j } \omega_{i,j}\ket{i,j} \otimes  \ket{\phi_{i,j}}
\end{equation*}
are  LU-equivalent to $\ket{\phi_{\alpha =0}}$, and hence pairwise equivalent. 
Indeed, the LU-equivalence is a composition of the aforementioned transformations. 
The matrices $U_1,\ldots, U_N$ is the simplest matrices satisfying restrictions given in \cref{verification}; which explains how they were found. 
Observe, that the assumption $2k<N$ was irrelevant in the presented analysis. Therefore, we conclude this discussion in the following corollary.

\begin{corollary}
\label{phases}
Two $2$-uniform states of minimal support which differs only with phases, i.e.
\begin{align*}
\ket{\psi} =&\sum_{I \in \mathcal{I}} \omega_I \ket{I} , \\
\ket{\psi '} =&\sum_{I \in \mathcal{I}} \omega_I ' \ket{I} ,
\end{align*}
where the sum runs over multi-index set $\mathcal{I} \subset [d]^N$ of size $| \mathcal{I} | =d^k$, are always LU-equivalent 
(and hence belong to the same SLOCC class).
\end{corollary}

\begin{example}
All AME(4,3) states $\ket{\text{AME(4,3)}_\omega} $ from \cref{eq2} are LU-equivalent, and hence they belong to the same SLOCC class. 
Similarly, the following states
\begin{equation*}
\ket{\text{AME(5,d)}_\omega} =  \dfrac{1}{d} \Bigg( 
\sum_{i,j=0}^{d-1} \omega_{i,j}  \ket{i,j,i+j,2i+j,3i+j}  \Bigg)
\end{equation*}
for $|\omega_{i,j} |=1$ are LU- and SLOCC-equivalent.
\end{example}

From \cref{phases}, it becomes clear that the diversity of possible phases in front of each term in the $2$-uniform state with minimal support does not reflect in the number of SLOCC classes. 
In fact, each $2$-uniform state with minimal support is equivalent to the one with all phases equal to $1$. 

Therefore, enumeration of SLOCC classes might be restricted to phases $1$ states only. 
In fact, it is equivalent to the classification of relevant OA. 



\begin{corollary}
Classification of $2$-uniform states with minimal support for $N>2k$ is equivalent to the classification of relevant OA, i.e. OA($d^k$,N,$d$,$k$) up to permutation of indices on each position. Potentially, for $N=2k$ two AME(4,d) states of minimal support might be in the same SLOCC class, even though the corresponding OA is not equivalent.
\end{corollary}

In literature, the classification of OAs is considered up to permutations of rows and columns \cite{HedayatIndex1OA,BushStudies}. 
Permutation of columns resembles, however, the physical operation of  exchanging subsystems. 
Therefore, by dividing the state space into SLOCC classes one should always indicate whether such operations are considered under the division \cite{FourQubits,FourQubits8}.

In particular, by the classification of OA, there exists at most one OA($d^k$,$N$,$d$,$k$) for $d=2,\ldots ,17$ for any number $N$ \cite{OAlib}. 
Hence, $2$-uniform state with minimal support and the local dimension $d=2,\ldots ,17$ are always SLOCC equivalent or SLOCC equivalent after permutation of parties. 
Nevertheless, we checked that for $N=4,5, d=4,5,6,7$, permutation of parties is not necessary for being SLOCC equivalent. 
We suppose this is true in general.

\begin{conjecture}
All $2$-uniform states of minimal support are LU-equivalent, and hence represent the same SLOCC class.
\end{conjecture}

In all verified cases there exist only one $2$-uniform state of minimal support. Nevertheless, some $2$-uniform states are not equivalent to the mentioned. In particular, for $d=5,7,11,13,\ldots$ there are AME(5,d) states of non-minimal support belonging to different SLOCC classes (see \cref{AME55}).

\subsection{3-uniform states}
\label{3-uniform states}

We have shown that the number of LU/SLOCC classes for $2$-uniform states of minimal support coincides with the number of relevant OAs which are non-isomorphic. 
In particular, two states which differ only with phases are always LU- and SLOCC-equivalent. 
This is in a strong contrast to the $3$-uniform states. 

\begin{example}
There exists AME(6,4) state with minimal support \cite{AME-QECC-Zahra}. 
In fact, this state might be obtained by reading consecutive rows of OA(64,6,4,3) from the OAs table \cite{OAlib}. 
Obviously, enhancing successive terms with any phase factor $|\omega | =1$ also yield to AME(6,4) state:
\begin{equation*}
\ket{\text{AME(6,d)}_\omega} =  \dfrac{1}{ d\sqrt{d}} \Bigg( 
\sum_{i,j,k=0}^{d-1} \omega_{i,j,k}  \ket{i,j,k} \otimes \ket{\psi_{i,j,k}}  \Bigg) .
\end{equation*}

Focus our attention on states with all phases $\omega_{i,j,k} =1$ with one exception: $\omega_{0,0,0}=\alpha$. 
Denote them as $\ket{\psi_\alpha}$ for simplicity. 
According to \cref{verification2}, the necessary condition for equivalence of such states $\ket{\psi_{\alpha_1}}$ and $\ket{\psi_{\alpha_2}}$ is
\[
\tfrac{ (W^{1,2}_{00})'}{W^{1,2}_{\sigma_1 (0), \sigma_2 (0)}} =\tfrac{(W^{1,2}_{0,1})'}{ W^{1,2}_{\sigma_1 (0), \sigma_2 (1)}}
\]
for any permutations $\sigma_1, \sigma_2$. 
According to the form of permutations, we have thus:
\begin{enumerate}
\item if $(\sigma_1 (0), \sigma_2 (0))=(0,0)$, then $\alpha_1=\alpha_2$;
\item if $(\sigma_1 (0), \sigma_2 (1))=(0,0)$, then $\alpha_1=\overline{\alpha_2}$;
\item otherwise $\alpha_1 ,\alpha_2 =1$.
\end{enumerate}
Therefore, if none of those conditions is satisfied, states $\ket{\psi_{\alpha_1}}$ and $\ket{\psi_{\alpha_2}}$ cannot be LU-equivalent. 
By the simple analysis, all states $\ket{\psi_{e^{i\phi}}}$ are pairwise non-LU-equivalent for $\phi \in [0,\pi )$.
\end{example}

Observe, that in such a way, we obtained a continuous family of non-LU-equivalent AME(6,4) states with minimal support. 
We conclude this observation in the corollary below. 
In fact, if the necessary conditions from \cref{verification2} are satisfied, the LU-equivalence may be provided (similarly to the case of $2$-uniform states).

\begin{corollary}
The AME(6,4) states:
\begin{align*}
\ket{\text{AME(6,d)}_{e^{i\phi}}} :=  
&\dfrac{1}{ d\sqrt{d}} \Bigg(
e^{i\phi} \ket{000000} + \\
&\sum_{i,j,k\neq (0,0,0)}  \ket{i,j,k} \otimes \ket{\psi_{i,j,k}}  \Bigg) .
\end{align*}
are pairwise in different LU- and SLOCC-classes for all phases $\phi \in [0,\pi )$.
\end{corollary}

Notice that for any $k$-uniform state with minimal support where $k>2$, similar construction of continuous non LU-equivalent family might be provided.

\begin{corollary}
\label{AMEminSUP6}
If there exists a $k$-uniform state with minimal support $\ket{\psi}$ where $k>2$:
\begin{equation*}
\ket{\psi} = \dfrac{1}{\sqrt{d^k}} \sum_{i_1,\ldots,i_k =0}^{d-1} \ket{i_1,\ldots,i_k} \otimes \ket{\psi_{i_1,\ldots,i_k}} ,
\end{equation*}
then the following family of $k$-uniform states:
\begin{align*}
\ket{\psi_{e^{i\phi}}} :=  
& \dfrac{1}{\sqrt{d^k}}  \Bigg(
e^{i\phi} \ket{0,\ldots,0}  \otimes \ket{\psi_{0,\ldots,0}}+\\
&\sum_{( i_1,\ldots,i_k )\neq (0,\ldots,0)}  \ket{i_1,\ldots,i_k} \otimes \ket{\psi_{i_1,\ldots,i_k}}   \Bigg) 
\end{align*}
is pairwise non LU- and SLOCC-equivalent for all phases $\phi \in [0,\pi )$.
\end{corollary}

\subsection{Summary}

We summarize shortly the number of non-SLOCC-equivalent AME states and AME states with minimal support in \cref{table1} and \cref{table2} respectively. 

The existence of AME states with minimal support for $N,d <8$ was analyzed \cite{Bernal75} based on the table of OAs and similar combinatorial designs. 
According to the discussion presented in the previous sections, if $N\geq 6$ existence of AME(n,d) state with minimal support persuade to infinitely many non-SLOCC-equivalent such states, see \cref{AMEminSUP6}.

Verification of the existence of AME states (not necessarily with minimal support) is far more complex problem. We refer to the tables of AME states \cite{AMElib} which summarizes several results concerning this problem \cite{HIGUCHI2000213,AME(4,Felix72,Bernal75,Huber_2018}. 
Even though the exact classification of AME states up to SLOCC-equivalence is yet unobtainable, in some specific cases non-trivial lower bound is given.

\begin{table}
\begin{tikzpicture}[x=1 cm, y=-1cm, node distance=0 cm,outer sep = 0pt]
\tikzstyle{day}=[draw, rectangle,  minimum height=1cm, fill=yellow!15,anchor=south west,text width=0.79cm,align=center]
\tikzstyle{day2}=[draw, rectangle,  minimum height=1cm, minimum width=1.5 cm, fill=yellow!20,anchor=south east,text width=1.79cm,align=center]
\tikzstyle{hour}=[draw, rectangle, minimum height=0.75 cm, minimum width=2 cm, fill=yellow!30,anchor=north east]
\tikzstyle{1hour}=[draw, rectangle, minimum height=0.75 cm, minimum width=1 cm, fill=yellow!30,anchor=north west]
\tikzstyle{Planche}=[2hour,fill=yellow!40,text width=0.8cm]
\tikzstyle{Planche1}=[1hour,fill=blue!10]
\tikzstyle{Planche2}=[1hour,fill=blue!25]
\tikzstyle{Planche3}=[1hour,fill=blue!40]
\tikzstyle{Planche0}=[1hour,fill=green!0]
\tikzstyle{Planche5}=[1hour,fill=blue!20]
\tikzstyle{Planche}=[1hour,fill=white]
\node[hour] (8-9) at (1,8) {AME(3,d)};
\node[hour] (9-10) [below = of 8-9] {AME(4,d)};
\node[hour] (10-11) [below= of 9-10] {AME(5,d)};
\node[hour] (11-12) [below = of 10-11] {AME(6,d)};
\node[hour] (12-13) [below  = of 11-12] {AME(7,d)};

\node[day] at (1,8)  {2 \scriptsize qubits};
\node[Planche1] at (1,8)  {1}; 
\node[Planche0] at (1,8.75) {0};
\node[Planche0] at (1,9.5) {0};
\node[Planche0] at (1,10.25)  {0};
\node[Planche0] at (1,11) {0};

\begin{scope}[shift={(1,0)}]
\node[day] at (1,8)  {3 \scriptsize qutrits};
\node[Planche1] at (1,8)  {1}; 
\node[Planche1] at (1,8.75) {1};
\node[Planche0] at (1,9.5) {0};
\node[Planche0] at (1,10.25)  {0};
\node[Planche0] at (1,11) {0};
\end{scope}

\begin{scope}[shift={(2,0)}]
\node[day] at (1,8)  {4};
\node[Planche1] at (1,8)  {1};
\node[Planche1] at (1,8.75) {1};
\node[Planche1] at (1,9.5) {1};
\node[Planche3] at (1,10.25)  {$\infty$}; 
\node[Planche0] at (1,11) {0};
\end{scope}

\begin{scope}[shift={(3,0)}]
\node[day] at (1,8)  {5};
\node[Planche1] at (1,8)  {1};
\node[Planche1] at (1,8.75) {1};
\node[Planche1] at (1,9.5) {1};
\node[Planche3] at (1,10.25)  {$\infty$}; 
\node[Planche0] at (1,11) {0};
\end{scope}

\begin{scope}[shift={(4,0)}]
\node[day] at (1,8)  {6};
\node[Planche1] at (1,8)  {1};
\node[Planche0] at (1,8.75) {0};
\node[Planche0] at (1,9.5) {0};
\node[Planche0] at (1,10.25)  {0}; 
\node[Planche0] at (1,11) {0};
\end{scope}

\begin{scope}[shift={(5,0)}]
\node[day] at (1,8)  {7};
\node[Planche1] at (1,8)  {1};
\node[Planche1] at (1,8.75) {1};
\node[Planche1] at (1,9.5) {1};
\node[Planche3] at (1,10.25)  {$\infty$}; 
\node[Planche3] at (1,11) {$\infty$};
\end{scope}

\node[day2] at (1,8) {local dimension};
\end{tikzpicture}
\caption{\label{table2} 
The exact number of not SLOCC-equivalent AME states with minimal support presented on a differently shaded blue background.}
\end{table}

\begin{table}
\begin{tikzpicture}[x=1 cm, y=-1cm, node distance=0 cm,outer sep = 0pt]
\tikzstyle{day}=[draw, rectangle,  minimum height=1cm, fill=yellow!15,anchor=south west,text width=0.79cm,align=center]
\tikzstyle{day2}=[draw, rectangle,  minimum height=1cm, minimum width=1.5 cm, fill=yellow!20,anchor=south east,text width=1.79cm,align=center]
\tikzstyle{hour}=[draw, rectangle, minimum height=0.75 cm, minimum width=2 cm, fill=yellow!30,anchor=north east]
\tikzstyle{1hour}=[draw, rectangle, minimum height=0.75 cm, minimum width=1 cm, fill=yellow!30,anchor=north west]
\tikzstyle{Planche}=[2hour,fill=yellow!40,text width=0.8cm]
\tikzstyle{Planche1}=[1hour,fill=blue!10]
\tikzstyle{Planche2}=[1hour,fill=blue!25]
\tikzstyle{Planche3}=[1hour,fill=blue!40]
\tikzstyle{Planche0}=[1hour,fill=green!0]
\tikzstyle{Planche5}=[1hour,fill=blue!20]
\tikzstyle{Planche}=[1hour,fill=white]
\node[hour] (8-9) at (1,8) {AME(3,d)};
\node[hour] (9-10) [below = of 8-9] {AME(4,d)};
\node[hour] (10-11) [below= of 9-10] {AME(5,d)};
\node[hour] (11-12) [below = of 10-11] {AME(6,d)};
\node[hour] (12-13) [below  = of 11-12] {AME(7,d)};

\node[day] at (1,8)  {2 \scriptsize qubits};
\node[Planche1] at (1,8)  {1}; 
\node[Planche0] at (1,8.75) {0};
\node[Planche1] at (1,9.5) {1};
\node[Planche1] at (1,10.25)  {1};
\node[Planche0] at (1,11) {0};

\begin{scope}[shift={(1,0)}]
\node[day] at (1,8)  {3 \scriptsize qutrits};
\node[Planche1] at (1,8)  {1}; 
\node[Planche1] at (1,8.75) {1};
\node[Planche1] at (1,9.5) {1};
\node[Planche1] at (1,10.25)  {1};
\node[Planche1] at (1,11) {1};
\end{scope}

\begin{scope}[shift={(2,0)}]
\node[day] at (1,8)  {4};
\node[Planche1] at (1,8)  {1};
\node[Planche1] at (1,8.75) {1};
\node[Planche1] at (1,9.5) {1};
\node[Planche3] at (1,10.25)  {$\infty$}; 
\node[Planche1] at (1,11) {1};
\end{scope}

\begin{scope}[shift={(3,0)}]
\node[day] at (1,8)  {5};
\node[Planche1] at (1,8)  {1};
\node[Planche1] at (1,8.75) {1};
\node[Planche2] at (1,9.5) {2};
\node[Planche3] at (1,10.25)  {$\infty$}; 
\node[Planche1] at (1,11) {1};
\end{scope}

\begin{scope}[shift={(4,0)}]
\node[day] at (1,8)  {6};
\node[Planche1] at (1,8)  {1};
\node[Planche0] at (1,8.75) {0?};
\node[Planche1] at (1,9.5) {1};
\node[Planche1] at (1,10.25)  {1}; 
\node[Planche0] at (1,11) {0?};
\end{scope}

\begin{scope}[shift={(5,0)}]
\node[day] at (1,8)  {7};
\node[Planche1] at (1,8)  {1};
\node[Planche1] at (1,8.75) {1};
\node[Planche2] at (1,9.5) {2};
\node[Planche3] at (1,10.25)  {$\infty$}; 
\node[Planche3] at (1,11) {$\infty$};
\end{scope}

\node[day2] at (1,8) {local dimension};
\end{tikzpicture}
\caption{\label{table1} 
The minimal number of non-SLOCC-equivalent AME states. The question mark by zero value suggests that the existence of the relevant state is dubitative, while $0$ itself emphasizes that the relevant state certainly does not exist.}
\end{table}

\section{Combinatorial designs}
\label{Combinatorial designs}

It is not our intention to provide a full picture of interactions between AME states and different combinatorial designs. 
For that purpose, we refer to the Goyeneche et al. (2018) \cite{DiK}, where the comprehensive introduction to that topic is presented. 
We shall, however, present some definitions directly related to our considerations.

In general, classical combinatorial designs (as orthogonal arrays; mutually orthogonal Latin squares, cubes, and hypercubes) are related to AME and $k$-uniform states of minimal support. 
Quantized versions of such a combinatorial designs are related to arbitrary AME and $k$-uniform states. 
Our work is restricted to minimal support states, hence the presentation of quantum combinatorial designs is not needed here.

We begin by introducing the necessary notation. 
Consider a discrete hypercube $[d]^k$ of dimension $k$. One can relate to $[d]^k$ the lower dimensional hypercube $[d]^s$ in two natural ways. 
Firstly, by choosing $k-s$ indices $S =\{s_1 ,\ldots , s_{k-s} \} \subset [k]$ and their values $i_1 ,\ldots ,i_{k-s} \in [d]$, there is an injective map
\[
[d]^s \cong [d]^k_{s_1 =i_1 ,\ldots , s_{k-s}  =i_{k-s}} 
\overset{i}{\hookrightarrow} [d]^d .
\]
Secondly, for any subset $S' \subset [k]$ of indices where $|S'|=s$, one can simply forget about indices out of $S'$. This operation is relevant to the surjection 
\[
[d]^k \overset{sur}{\longrightarrow}  [d]^k_{| S'} \cong[d]^s .
\]

\begin{definition}
\label{MOLH}
A $k$-\emph{mutually orthogonal Latin hypercubes} MOLH of size $d$ and dimension $k$ is a bijection 
\[
L :[d]^k \longrightarrow [d]^k 
\]
such that by choosing any set  $S =\{s_1 ,\ldots s_{k-s} \} \subset [k]$ of $k-s$ indices and their values $i_1 ,\ldots ,i_{k-s} \in [d]$, and any subset $S' \subset [k]$, the composition of $L$ with above defined injection $i$ (on the left) and surjection $sur$ (on the right) provides a bijection:
\[
[d]^s \cong [d]^k_{s_1 =i_1 ,\ldots , s_{k-s}  =i_{k-s}} 
\xrightarrow{i \circ L \circ sur}
[d]^k_{| S'} \cong[d]^s .
\]
We denote such an object as $k$-MOLH($d$).
\end{definition}

\begin{example} 
Bijection
\[
L :[d]^2 \longrightarrow [d]^2 
\]
such that in each row and on each position, all elements appear exactly once constitutes a \emph{mutually orthogonal Latin square} MOLS($d$). Here \emph{square} stands for MOLH dimension $k=2$.
\end{example}

In general, orthogonality and dimension of MOLH might be indexed by different numbers (here both are equal and denoted by $k$) \cite{ComDesi}. 
This distinction is, however, not needed for our purpose. 

There is a one-to-one correspondence between $k$-MOLH(d) and orthogonal arrays OA($d^k ,k,2k,k$), and hence between them and AME(2k,d) states of minimal support \cite{ComDesi}.

\begin{proposition}
\label{hcc}
Any AME($2k$, $d$) state of minimal support is equivalent to $k$-mutually orthogonal Latin hypercube $L$:
\[
L(i_1 ,\ldots , i_k) := \big({\phi_{I}^1},\ldots, {\phi_{I}^k} \big) .
\]
\end{proposition}

\begin{proof}
Consider a $k$-MOLH(d) L. 
By adjusting 
\begin{equation}
\label{eq100}
i_1,\ldots ,i_k , {\phi_{I}^1},\ldots, {\phi_{I}^k}
\end{equation}
into $d^k $ rows (with $2k$ elements each), one obtains OA($d^k ,k,2k,k$). Indeed, choose any set of $k$ indices and split it into two: $S  \cup S'$, where $S$ is a $(k-s)$-elementary subset of the first half of indices, and $S'$ is an $s$-elementary subset of the second half of indices. 
For any choice of values $i_1 ,\ldots ,i_{k-s} \in [d]$, by \cref{MOLH}, there is a bijection
\[
[d]^k_{s_1 =i_1 ,\ldots , s_{k-s}  =i_{k-s}} 
\xrightarrow{i \circ L \circ sur}
[d]^k_{| S'} ,
\]
and hence the subset $S  \cup S'$ of k columns in \cref{eq100} contains all possible combinations of symbols. 
Since the choice of $k$-elementary subset $S  \cup S'$ was unrestricted, this is a defining property of OA of index unity. 
Overturning this argument provides the reverse statement.
\end{proof}

\begin{example}
The AME(4,3) state from \cref{AME43p}  is equivalent to the mutually orthogonal Latin square (2-MOLH):
\begin{center}
\begin{tikzpicture}
\tikzset{%
square matrix/.style={
    matrix of nodes,
    column sep=-\pgflinewidth, 
    row sep=-\pgflinewidth,
    nodes in empty cells,
    nodes={draw,
      minimum size=#1,
      anchor=center,
      align=center,
      inner sep=0pt
    },
    column 1/.style={nodes={fill=cyan!20}},
    row 1/.style={nodes={fill=cyan!20}},
  },
  square matrix/.default=0.8cm
}

\matrix[square matrix] (A)
{
  & 0 & 1 & 2 \\
0 & 00& 11&22 \\
1 & 12 & 20 &   01 \\ 
2 & 21 & 02& 10\\ 
};

\draw (A-1-1.north west)--(A-1-1.south east);
\node[below left=2mm and 4mm of A-1-1.north east] {$i$};
\node[above right=2mm and 4mm of A-1-1.south west] {$j$};
\end{tikzpicture}
\end{center}
\noindent
The entries of MOLS are pairs of numbers $(k,\ell)$. The relevant quantum state is obtained by reading all entries:
\[
\ket{\text{AME(4,3)}} =\dfrac{1}{d} \sum_{i,j=0}^{d-1}\ket{i,j,k,\ell}.
\]
\end{example}

\noindent
The notion of MOLS was used for the construction of several AME states \cite{AME(4}.



\subsection{Existence  of Latin designs}

\begin{proposition}
\label{prop5}
If a mutually orthogonal Latin hypercube $k$-MOLH($d$) for $k>1$ exists, then indices $d,k$ satisfies
\[
k\leq d-1 .
\]
\end{proposition}

\begin{proof}
The hyper-row $I:=(i,0,\ldots,0)$, $i\in d$ defines the following mapping:
\[
L_{|I} : i \longmapsto \Big( {\phi_{i}^1},\ldots, {\phi_{i}^k} \Big) \in [d]^k
\]
where $\phi_{i}^\ell := \phi_{I}^\ell $ for simplicity. Observe that on each position all symbols appear, i.e.
\begin{equation}
\label{eq101}
\left\lbrace \phi_{i}^\ell : i \in [d] \right\rbrace  =[d] .
\end{equation} 

Consider now the element 
\[
(j_1 , \ldots,j_k) := L(0,1,0,\ldots,0) \in [d]^k.
\]
From \cref{eq101}, clearly
\[
j_1 = \phi_{i_1}^1 ,\quad \ldots , \quad j_k = \phi_{i_k}^k
\]
for some indices $i_1,\ldots ,i_k$. 
Observe that
\begin{enumerate}
\item $i_1,\ldots ,i_k \neq 0$. Suppose the contrary, $i_s =0$. 
Then $j_s = \phi_{0}^s$, and hence 
\begin{align*}
L(0,0,0,\ldots,0) &= (\ldots ,j_s ,\ldots ), \\
L(0,1,0,\ldots,0) &= (\ldots ,j_s ,\ldots ) ,
\end{align*}
where dotted symbols on the right are not specified. This is in contradiction to \cref{MOLH} for $S$ given by $i_j =0$ for all $j \neq 1$, and $S' =\{s \}$;
\item indices $i_1,\ldots i_k$ are pairwise different. Suppose the contrary, $i_{s_1}=i_{s_2}$. Then
\begin{align*}
L(i_{s_1},0,0,\ldots,0) &= (\ldots ,j_{s_1},\ldots ,j_{s_2} ,\ldots ) ,\\
L(0,1,0,\ldots,0) &= (\ldots ,j_{s_1},\ldots ,j_{s_2} ,\ldots ).
\end{align*}
This is in contradiction to \cref{MOLH} for $S$ given by $i_j =0$ for all $j \neq 0,1$, and $S' =\{s_1,s_2 \}$.
\end{enumerate}
Since all indices $i_1,\ldots ,i_k \in [d]$ are pairwise different and non-zero, $k \leq d-1$.
\end{proof}

It is worth mentioning that the condition given in \cref{prop5} is only a necessary condition for the existence of MOLH. 
If it is satisfied, the precise construction of MOLH is known for all $d$ being prime powers. 
This construction might be extended further by composing two MOLH of a different size.
Nevertheless, the aforementioned condition is not a sufficient one. 
For instance, construction of MOLS ($2$-MOLH) of size $d=6$ refers to the famous problem of \emph{36 officers of Euler} \cite{Euler36}, which was proven to have no solution \cite{GastonTarry}. 

\subsection{Extension of Latin designs}

As we shall see, Latin designs are not only related to the construction of AME states with minimal support, but also to the local unitary relations between such. 
In particular, the problem of existence and extension of k-dimension MOLH is relevant to the description of LU-equivalences between AME(2k,d) states. 
Therefore, the short outline of the extension problem is presented below.

\begin{definition}
A MOLH($s$) of size $k$ and dimension $s$ 
\[
L :[s]^k \longrightarrow [s]^k 
\]
might be extent to MOLH($d$) if there exists MOLH($d$):
\[
L' :[d]^k \longrightarrow [d]^k ,
\]
which preserves the structure of $L$, i.e. $L_{\big\vert [s]^k}' \equiv L$. Moreover, we refer to $L$ as a sub-MOLH($s$) of MOLH($d$).
\end{definition}

\begin{example}
A MOLS(3) $L$ might be extent into MOLS(9) presented below. 
\begin{center}
\begin{tikzpicture}
\tikzset{%
square matrix/.style={
    matrix of nodes,
    column sep=-\pgflinewidth, 
    row sep=-\pgflinewidth,
    nodes in empty cells,
    nodes={draw,
      minimum size=#1,
      anchor=center,
      align=center,
      inner sep=0pt
    },
    column 1/.style={nodes={fill=cyan!20}},
    row 1/.style={nodes={fill=cyan!20}},
  },
  square matrix/.default=0.7cm
}

\matrix[square matrix] (A)
{
  & 0&1 & 2 & 3 & 4 & 5 &6&7&8\\ 
0& 00& 11&22 & 33& 44&55&66&77&88 \\ 
1 & 12& 20& 01& 45& 53& 34&78& 86& 67 \\ 
2 & 21 & 02 & 10& 54 & 35 & 43 &87 & 68 & 76 \\ 
3&  36& 47&58&60&71&82& 03& 14&25  \\ 
4 & 48& 56& 37&72& 80& 61&15& 23& 04\\ 
5 & 57 & 38 & 46 &81 & 62 & 70& 24 & 05 & 13 \\ 
6&63&74&85& 06& 17&28 & 30& 41&52 \\ 
7 &75& 83& 64& 18& 26& 07& 42& 50& 31 \\ 
8 &84 & 65 & 73 & 27 & 08 & 16& 51 & 32 & 40 \\ 
};

\draw (A-1-1.north west)--(A-1-1.south east);
\node[below left=2mm and 3.5mm of A-1-1.north east] {$i$};
\node[above right=2mm and 3.5mm of A-1-1.south west] {$j$};

\draw[line width=1.5pt](A-2-1.north west)--(A-2-10.north east);
\draw[line width=1.5pt](A-5-1.north west)--(A-5-10.north east);
\draw[line width=1.5pt](A-8-1.north west)--(A-8-10.north east);
\draw[line width=1.5pt](A-10-1.south west)--(A-10-10.south east);

\draw[line width=1.5pt](A-1-2.north west)--(A-10-2.south west);
\draw[line width=1.5pt](A-1-5.north west)--(A-10-5.south west);
\draw[line width=1.5pt](A-1-8.north west)--(A-10-8.south west);
\draw[line width=1.5pt](A-1-10.north east)--(A-10-10.south east);
\end{tikzpicture}
\end{center}
Indeed, one can see that entries in the square consisting of three first rows and columns are taken from $0,1,2$ only. 
In fact, this extension is relevant to a tensor product of two identical Latin squares $L$.
\end{example}


\begin{remark}
\label{ob3}
Consider a MOLH($d$) $L$. 
If there exist a sub-hypercube $S =S_1 \times \cdots \times S_k \subseteq [k]^d$ which is mapped by $L$ on another hypercube $S' =S_1' \times \cdots \times S_k' \subseteq [k]^d$, then up to permutation of labels, $L':= L_{\big\vert S} $ is sub-MOLH($s$) of MOLH($d$). 
Moreover, $L$ cannot map hyper-rectangles onto hyper-rectangles except  hypercubes into hypercubes. 
Hence, the notion of sub-hyper-rectangle of $L$ does not exists.
\end{remark}

The problem of extension of Latin designs might be traced by to the Ryser’s Theorem \cite{Ryser}, and is a plentiful scientific problem considered in several papers \cite{Extension1,Extension2}. 
It is not our intention to provide a full picture of interactions between AME states and different combinatorial designs. 
We shall show, however, what are the dimension bounds for extending MOLH($s$) into MOLH($d$); or equivalently finding sub-MOLH($s$) of MOLH($d$).

\begin{proposition}
\label{prop4}
Inequality
\[
s \leq \dfrac{1}{1 +\sqrt[k-1]{k}} \enskip d
\]
is a necessary condition for extension MOLH($s$) into MOLH($d$) for any $k>1$.
\end{proposition}

\begin{proof}
Suppose that $L'$ extent $L$.  Consider pairwise disjoint sets:
\[
S_i := \underbrace{S \times S }_{i-1} \times S^c \times \underbrace{S \times S }_{k-i} 
\]
where $S:=[s]$. 
Denote their sum by $S'=\cup_{i=1}^k S_i$. 
Observe that for any multi-index $I \in S'$:
\[
L(I) \in S^c \times\cdots\times S^c .
\]
Indeed, it follows from the fact that $L(S) =S$, and hence on any position no indices from $S$ might appear in $L(S')$. 
Since $L$ is the bijection:
\[
\vert S'\vert \leq\vert  S^c \times\cdots\times S^c \vert ,
\]
and hence
\[
k s^{k-1}  \leq \big( d-s \big)^{k-1},
\]
which is equivalent to the statement of \cref{prop4}.
\end{proof}

As we shall see, the existence of non-trivial sub-MOLH(s) in MOLH(d) is directly related to the problem of describing the automorphisms of AME(2k,d) states. More precisely, assumptions that either non-trivial sub-MOLH(s) do not exists or cannot be extended to MOLH(d) allows to provide a comprehensive description of LU-equivalences of AME(2k,d) states. 
Therefore, the necessarily conditions of existence and extension MOLH are limiting the statements of \cref{prop1=} and \cref{verification2}. This limitation is notified in \cref{small}, which follows directly from \cref{prop5} and \cref{prop4}.



\section{Conclusions}
\label{Conclusions}

In this paper, we develop techniques of SLOCC-verification between $k$-uniform and AME states. 
In particular, we show that two $k$-uniform states are SLOCC-equivalent iff they are locally monomial (LM)-equivalent. 
We further specify the matrices which might appear in such equivalences. 
These results significantly restrict the class of possible local transformations to a finite set, which makes SLOCC-verification feasible. 

For AME(2k,d) states, the aforementioned statement is not true anymore. 
Intriguingly, SLOCC-equivalences might be provided by Fourier transforms, and in general, by Butson-type matrices. 
This restriction is valid, however, only for small local dimensions $d$ and number of parties $N$ (in particular for arbitrary $N$ and $d<9$). 
The exact bound on $d$ and $N$ is related to the necessary condition for existence and extension of combinatorial designs called mutually orthogonal hypercubes. 
Despite the exhaustive analysis performed, the general structure of SLOCC-equivalences between AME(2k,d) states is still puzzling and remains unknown. 
We present evidence that exceeding this class of equivalences is possible only in composed systems. 
General results concerning SLOCC-equivalences of AME(2k,d) states are also presented. 

We illustrate the usefulness of the provided criteria on various examples. 
Firstly, we show that the existence of AME states with minimal support of 6 or more particles yield the existence of infinitely many such non-SLOCC-equivalent states. 
The exact number of SLOCC-classes containing AME states with minimal support is given. 
Secondly, we show that some AME states cannot be locally transformed into existing AME states of minimal support. 
This shows that the notion of support is relevant even for AME states. 

\section{Further discussion and open problems}
\label{Further discussion and open problems}


Ultimate LU- and SLOCC-classification of $k$-uniform states, even of minimal support, is in fact a complex project involving many open mathematical problems, such as:
\begin{enumerate}
\item existence and extension of mutually orthogonal Latin hypercubes, see \cref{Combinatorial designs},
\item classification of Hadamard matrices of Butson type B(d,d),
\item classification/uniqueness of OAs of index unity (without permutation),
\end{enumerate}
among others. 
Therefore, with a great conviction, we claim it to be currently out of reach. 
Below, we discuss three open problems regarding LU- and SLOCC-classification of $k$-uniform states with minimal support in a detailed way. 
We show their connections with some open mathematical problems.

Firstly, consider two $k$-uniform states of minimal support $\ket{\psi}$ and $\ket{\psi '}$ with all phases equal to $1$ for simplicity. 
With this constrain on phases, \cref{prop1} shows that $\ket{\psi}$ and $\ket{\psi '}$ are LU-equivalent if and only if there exist local permutation matrices relating $\ket{\psi}$ and $\ket{\psi '}$:
\[
\ket{\psi '} = \sigma_1 \otimes\cdots\otimes\sigma_n \ket{\psi} .
\]
States $\ket{\psi}$ and $\ket{\psi '}$ are in one-to-one correspondence with two OAs of index unity. 
The existence of local permutation matrices is equivalent to an isomorphism between two OAs of index unity. 
Hence LU-classification of such a states is equivalent to the classification of OAs of index unity. 
Such a classification is, however, an open mathematical problem. 
In many situations, when number of parties $N$, uniformity $k$, local dimension $d$ is small, it is known that all OAs of index unity are isomorphic \cite{BULUTOGLU2008654,doi:10.1080/00401706.1992.10484952,stufken2007}. 

\begin{conjecture}
\label{productCon}
All OAs of index unity are isomorphic by permutations of symbols on each level. 
Equivalently, all $k$-uniform states with minimal support and all terms phases equal are LU-equivalent.
\end{conjecture}

Secondly, in \cref{verification2,prop1=}, the form of arbitrary LU-operator between two AME(2k,d) states with minimal support is provided for small numbers $k$ and $d$. 
It is given by a Butson-type matrix B(d,d) or an identity matrix, multiplied by local monomial matrices form both sides. 
We have shown that for a composed system there are local operators beyond the provided formula. 
Indeed, in this case, the tensor product of Butson-type matrix and the identity matrix may provide LU-equivalence. 
We conjecture that it is a general form for LU-equivalences for all AME(2k,d) states, and it is tightly related to the possible decomposition of a system. 
This supposition is stated in \cref{conProd}. 

Thirdly, \cref{prop1} states that any LU-operator between two $k$-uniform states of minimal support $\ket{\psi}$ and $\ket{\psi '}$ is a local product of phase (diagonal) and permutation matrices (for $2k<N$). 
By considering states $\ket{\psi}$ and $\ket{\psi '}$ with terms of various phases, we showed that not all of them are LU-equivalent for $k>2$. 
Nevertheless, the precise description of SLOCC classes containing such states is not given. 
Therefore, the role of permutation matrices in LU-classification is not yet absolutely clear. 

Finally, in \cref{32} the basic difference between $k$-uniform states of minimal support where $2k<N$ and $2k=N$ is discussed. LU-equivalence between two $k$-uniform states with $2k=N$ decomposes into multiplication of Butson-type matrix and local monomial (LM) matrices from both sides. 
Obviously, Butson-type matrices significantly increase the class of LU-equivalences between two states. Nevertheless, it is not known yet whether such LU-equivalences are beyond local monomial equivalences. 
In fact, in all provided examples involving Butson-type matrices in LU-equivalence, states were always LM-equivalent. 
Therefore we conjecture that \cref{coro1} holds true in the case $2k =N$ (even though \cref{prop1} does not hold anymore). 

\begin{conjecture}
All AME(2k,d) states with minimal support are LU-equivalent if and only if they are LM-equivalent
\end{conjecture}

Notice that any attempt of proving the statement above makes sense only if \cref{productCon} would be true. In such a case, the classification Butson type B(d,d), which is an open mathematical problem, and refine analysis of such is required. 

\section{Acknowledgements}
Authors are thankful to Karol \.{Z}yczkowski, Felix Huber, Gon\c{c}alo Quinta, Wojciech Bruzda and all other collaborators and colleagues for valuable and fruitful discussions, which greatly improved this text. 
AB acknowledges support from the National Science Center under grant number DEC-2015/18/A ST2/00274. 
ZR acknowledges support from the Spanish MINECO (Severo Ochoa SEV2015-0522), Fundacio Cellex and Mir-Puig, Generalitat de Catalunya (SGR 1381 and CERCA Programme), and ERC AdG CERQUTE.

\appendix

\section{The proof of \cref{prop1}}
\label{app1}

We shall prove \cref{prop1} in a slightly enhanced version. 
Notice that LU- and SLOCC-equivalences coincides on the class of AME states, which is an immediate conclusion from \cref{LU=SLOCC}. 
Therefore, we restrict our argument to LU-equivalences only. 
We would like to emphasize the statement below as primary and more valuable for the LU-verification procedure then the claim of \cref{prop1} itself. 
Indeed, this extended version is used later in \cref{AppD} for the demonstration of non-equivalence of two families of AME(5,d) states.

\begin{proposition}
\label{prop2}
Consider two $k$-uniform states $\ket{\psi}$ and $\ket{\phi}$ with minimal support. 
For any subsystem $S$ consisting of $s>k$ parties, the reduced density matrices $\rho_{ S} \left( \psi \right) $ and $\rho_{ S} \left( \phi \right)$ 
are LU-equivalent if and only if they are LM-equivalent.
\end{proposition}

Observe that \cref{prop1} is an immediate consequence of the statement above for subsystem $S$ of all parties, i.e. $|S|=N$. 
Without loss of generality it is enough to prove the statement of \cref{prop2} only for the smallest possible subsystems $S$, i.e. consisting of $k+1$ parties. 
Indeed, assume that the reduced states $\rho_{ S'} \left( \psi \right) $ and $\rho_{ S'} \left( \phi \right)$ are equivalent by a local unitary matrix $U $. 
Consider any subsystem $S \subseteq S'$ of $k+1$ parties. 
The local operator $U$ splits:
\begin{equation*}
U = U_S U_{S'\setminus S}
\end{equation*}
\noindent
where $U_S$ is the local operation on $S$ subsystem; and $U_{S'\setminus S}$ on $S' \setminus S$ subsystem equivalently. 
Since the (partial) trace is invariant under cycling permutations, we have:
\[
\rho_{S } \left( \psi \right) =  U_{ S} \Big( \rho_{S } \left( \phi  \right)  \Big)
\]
and by \cref{prop2}, $U_{ S}$ is a local monomial operation. 
Since the subsystem $S$ was chosen arbitrary, $U$ is a local monomial operator. 

Notice that the size $s>k$ of the subsystem $S$ in \cref{prop2} is the largest possible. 
Indeed, after taking the partial trace over larger subsystem, both states $\ket{\psi}$ and $\ket{\phi}$ become proportional to the identity, and hence any local unitary operation provides their equivalence. 

We introduce the following notation.  
Consider two LU-equivalent $k$-uniform states: $\ket{\psi}$, $\ket{\phi}$ of minimal support form. We make use of the decomposition into supports elements:
\[
\ket{\psi} =\sum_{i=1}^{d^k} \alpha_i \ket{\psi_i}
\]
where $\ket{\psi_i}$ are of the unity support. Moreover, we denote elements of $\ket{\psi_i}$ as follows:
\[
\ket{\psi_i} =\ket{x^i_1 \cdots x^i_N} 
\]
where $x^i_j =0,\ldots ,d-1$. 
We use the similar notation for the state $\ket{\phi}$:
\[
\ket{\phi} =\sum_{i=1}^{d^k} \alpha_i \ket{\phi_i}, 
\quad\quad 
\ket{\phi_i} =\ket{y^i_1 \cdots y^i_N} .
\]

\begin{lemma}
\label{lem}
For a partial trace over any subsystem $S$ of $|S| \geq k$ parties 
\[
\tr_S \ket{\psi}\bra{\psi} = \sum_{i=0}^{d^k} 
\ket{\widetilde{\psi_i}}\bra{\widetilde{\psi_i}}
\]
where $\ket{\widetilde{\psi_i}} =\tr_S \ket{\psi_i}$. Moreover, for any $i \neq j$, $\ket{\psi_i}$ and $\ket{\psi_j}$ coincides on at most $k-1$ positions.
\end{lemma}

\begin{proof}
For any subsystem $S$ such that $|S|=N-k$, $\tr_{S} \ket{\psi}\bra{\psi} = \Id_{d^k}$. 
Hence, for any $i \neq j$, vectors $\ket{\psi_i}$ and $\ket{\psi_j}$ coincide on at most $k-1$ positions. 
Indeed, suppose the contrary, i.e. they coincide on some $k$ positions. Then, by tracing out the rest of parties, we get $\tr_{N-k} \ket{\psi_i}\bra{\psi_i} = \tr_{N-k} \ket{\psi_j}\bra{\psi_j}$ and by the minimality of the support $\tr_{N-k} \ket{\psi}\bra{\psi} \neq \Id_{d^k}$. 
The statement of the lemma follows immediately form presented observation.
\end{proof}

\begin{proof}[Proof of \cref{prop2}]
As we already discussed, it is enough to prove \cref{prop2} for subsystems $S$ of $k+1$ parties. 
Without loss of generality, analyze the reduction to the subsystem $S$ of the first $k+1$ parties.
Consider a local unitary operation $U:= U_1\otimes \cdots \otimes U_{k+1}$ transforming $\rho_{S} \left( \psi \right)$ into $\rho_{S} \left( \phi \right)$. Unitary operators $U_i$ might be seen as the following change of basis:
\begin{equation}
\label{uff5}
\ket{\widehat{i_1}\cdots \widehat{i_{k+1}}} :=
\ket{U_1 ( {i_1})\cdots U_{k+1}( i_{k+1})} .
\end{equation}
From \cref{lem} we have:
\begin{equation}
\rho_{S} \left( \psi \right)= \sum_{i=1}^{d^k} 
\ket{x^i_1 \cdots x^i_{k+1}} \bra{x^i_1 \cdots x^i_{k+1}} ,
\end{equation}
and 
\begin{equation}
\label{uff}
\rho_{S} \left( \phi \right)= \sum_{i=1}^{d^k} 
\ket{y^i_1 \cdots y^i_{k+1}} \bra{y^i_1 \cdots y^i_{k+1}} .
\end{equation}
Observe that
\[
U \rho_{S} \left( \psi \right) U^{-1}= \sum_{i=1}^{d^k} 
\ket{\widehat{x^i_1} \cdots \widehat{x^i_{k+1}}} \bra{\widehat{x^i_1} \cdots \widehat{x^i_{k+1}}}  ,
\]
and since $U$ is the LU-equivalence, two expressions above are equal. Consequently, we have equality of the following spaces:
\begin{align*}
span &\Big\lbrace\ket{y^i_1 \cdots y^i_{k+1}} ; i =1,\ldots ,d^k\Big\rbrace = \\
&span \Big\lbrace \ket{\widehat{x^i_1} \cdots \widehat{x^i_{k+1}}} ; i =1,\ldots ,d^k\Big\rbrace .
\end{align*}
In general, each vector from the first space is a linear combination of vectors from the second space. We will show, however, that there is a one-to-one correspondence between vectors from both spaces; namely, for any index $i$ there exists an index $j_i$ such that:
\begin{align}
\label{uff4}
\ket{\widehat{x^i_1} \cdots \widehat{x^i_{k+1}}}  = \ket{y^{j_i}_1 \cdots y^{j_i}_{k+1}}
\end{align}
With this observation at hand, and by \cref{uff5},  the statement of \cref{prop1} follows immediately. 

What remains to show is that, indeed, \cref{uff4} holds. Consider the vector 
$\ket{\widehat{x^i_1} \cdots \widehat{x^i_{k+1}}}$ and present it as the following linear combination 
\begin{equation}
\label{uff2}
\ket{\widehat{x^1_1} \cdots \widehat{x^1_{k+1}}} = \sum_{i=1}^{d^k} \beta_i \ket{y^i_1 \cdots y^i_{k+1}}.
\end{equation}
On the other hand,
\begin{align}
\label{uff3}
\ket{\widehat{x^1_1} \cdots \widehat{x^1_{k+1}}} 
&=    \ket{U_1 ( x^1_1)\cdots U_{k+1}( x^1_{k+1})}  \\
&=\sum_{j_1 ,\ldots, j_{k+1} =0}^{d-1} u^1_{x^1_1 j_1} \cdots u^{k+1}_{x^1_{k+1}j_{k+1}} 
\ket{j_1 \cdots j_{k+1}}
\end{align}
where $u^k_{lm} = \left( U_k\right)_{lm} $, and hence
\[
\beta_i = u^1_{x^1_1  y^i_1 } \cdots u^{k+1}_{x^1_{k+1} y^i_{k+1}} .
\]

Suppose now, that for $i \neq j$, $\beta_i ,\beta_j \neq0$. Consequently,
\[
u^m_{x^1_m  y^i_m } \neq 0, u^m_{x^1_m  y^j_m } \neq 0,
\] 
for $m=1,\ldots, k+1$. By \cref{lem}, $\ket{y^i_1 \cdots y^i_{k+1}}$ and $\ket{y^j_1 \cdots y^j_{k+1}}$ differ on at least $2$ positions; without loss of generality suppose $y^i_1 \neq y^j_1 $. Observe that
\[
u^1_{x^1_1  y^i_1 }  \neq 0 ,u^2_{x^1_j  y^j_2 } \neq 0, \ldots, u^{k+1}_{x^1_{k+1} y^j_{k+1}} \neq 0
\]
and hence the expression
\[
\ket{y^i_1 y^j_2 \cdots y^j_{k+1}} \bra{y^i_1 y^j_2 \cdots y^j_{k+1}}
\]
appear on the right-hand side of \cref{uff} with non-zero coefficient, which is in contradiction to \cref{lem}.

Since there is at most one $\beta_i \neq 0$, the sum in \cref{uff2} collapse to the one term. Similar reasoning shows, that $\ket{\widehat{x^m_1} \cdots \widehat{x^m_{k+1}}} $ is in general equal to $\ket{y^i_1 \cdots y^i_{k+1}}$ for some $i=1,\ldots,d^k$, what should have been shown.
\end{proof}

Notice that the presented argument does not hold if $2k=N$. 
Indeed, the smallest non-trivial reduced system of $k$-uniform state consists of $k+1$ parties. 
The proof is based on \cref{lem} to justify that vectors $\ket{y^i_1 \cdots y^i_{k+1}}$ and $\ket{y^j_1 \cdots y^j_{k+1}}$ differ on at least $2$ positions. 
For $2k=N$, however, the argument of \cref{lem} might be used only for the trivial reduction to $k$ parties.

\section{The proof of \cref{verification}}
\label{verification2APP}

\begin{proof}[Proof of \cref{verification}] 
From \cref{prop1} the LU-equivalence between two states of minimal supports is a product of permutation and diagonal matrices. Suppose that the local permutation $\sigma$ was already applied to the state $\ket{\psi}$. Therefore, we may assume that $\ket{\psi}$ and $\ket{\psi '}$ are related only by diagonal operators.

Since both states $\ket{\psi}$ and $\ket{\psi '}$ are $k$-uniform states of minimal support related only by diagonal operators, they might be written in the following form:
\begin{align}
\ket{\psi}  &=\sum_{I \in\mathcal{I}} \omega_I \ket{I} , \\
\ket{\psi '}  &=\sum_{I \in \mathcal{I}} \omega_I ' \ket{I} \,
\end{align}
where $\mathcal{I} \subset [d]^n$ has dimension $|\mathcal{I} |= [d]^k$. 
Denote by
\[
D^\ell =\text{diag} (u^\ell_1 ,\ldots , u^\ell_d )
\]
the diagonal operators relating $\ket{\psi}$ and $\ket{\psi '}$. 
Clearly,
\begin{equation}
\label{eq-1}
\omega_I ' =u^1_{i_1} \cdots u^n_{i_n} \omega_I
\end{equation}
for any index $I={i_1} ,\ldots , {i_n}$. 
Denote by $U^\ell := \prod_{i \in [d]} u^\ell_i$ the product of all non-zero elements from the matrix $D^\ell$.

For simplicity, let us choose $S =\{1,\ldots ,k-1\}$ being the set of first $k-1$ indices. 
We shall show the statement with respect to the first matrix $D^1$.  
Consider any multi-index $I=i_2,\ldots ,i_{k-1}$. 
By multiplying adequate expressions from \cref{eq-1} by sides, one may obtain
\[
(W^S_{i,I}) '= \Big(u^1_i \Big)^d 
\Big(u^2_{i_2} \Big)^d \cdots \Big(u^{k-1}_{i_k-1} \Big)^d
 \Big( U^k\cdots U^n \Big)
 W^S_{i,I} .
\]
The fact that $U^k\cdots U^n $ appears on the right-hand side follows from the basic properties of OA related to the states $\ket{\psi}$ and $\ket{\psi '}$. 
Hence
\begin{equation}
\label{eq-2}
\Big(u^1_0 \Big)^d \dfrac{W^S_{0,I}}{(W^S_{0,I}) '} 
 =\cdots
= \Big(u^1_{d-1} \Big)^d \dfrac{W^S_{d-1, I}}{(W^S_{d-1,I}) '} .
\end{equation}
From this immediately follows that 
\[
D^1 =\omega_1 \text{diag} \Bigg( \sqrt[\leftroot{-3}\uproot{3}d]{\tfrac{(W^S_{0,I})'}{W^S_{0,I}}},
\ldots ,
\sqrt[\leftroot{-3}\uproot{3}d]{\tfrac{(W^S_{d-1,I})'}{W^S_{d-1,I}}} \Bigg),
\]
for some phase factor $\omega_1$. 

Since the multi-index $I$ is arbitrary, from \cref{eq-2} follows that
\begin{equation}
\tfrac{(W^S_{0,I})'}{W^S_{0,I}} =\tfrac{(W^S_{0,I'})'}{W^S_{0,I'}}
\label{promp}
\end{equation}
for any other multi-index $I'=i_2',\ldots ,i_{k-1}'$. 

Since we assumed that the local permutation $\sigma$ was already applied to the state $\ket{\psi}$, one has to consider the action of $\sigma$ on denominator in \cref{promp} in the general case. 
This shows the statement of \cref{verification} for the set $S $ of first $k-1$ indices. 
The same reasoning might be applied for other set $S$, and further for any indices: $2,\ldots, n$. The global phase is then a multiplication of obtained factors $\omega_i$. 
\end{proof}

\section{\cref{32} revisited}
\label{The proof Case II}

We discuss in details LU-equivalences of AME(2k,d) states with minimal support. In particular, proofs of \cref{prop1=,verification2,cor2} are presented. 
We begin with the necessary notation. 
For simplicity, all AME(2k,d) states considered in \cref{The proof Case II} are normalized to $\sqrt{d^k}$. 
In such a way, all terms in the computational basis of states with minimal support are normalized to one. 
For a local unitary operator $U:= U_1 \otimes \cdots \otimes U_\ell$ and a subset of indices $S \subseteq [\ell ]$, we define its $S$ part by:
\[
U_S :=\bigotimes_{i \in S} U_i.
\]

\begin{observation}
\label{Ob1}
Let $U:= U_1 \otimes \cdots \otimes U_\ell$ be a local unitary operator transforming a state $\rho (\psi ) \in \mathcal{H}^{\otimes \ell} $ onto $\rho (\psi ')$, i.e.
\[
U \Big(  \rho (\psi )\Big) =  \rho (\psi ' ) .
\]
Then, for any subset $S \subseteq [ \ell ]$ of indices: 
\begin{equation}
\label{eq7}
U_S  \Big(  \rho_S (\psi )\Big) =  \rho_S (\psi ' ) .
\end{equation}
Moreover, if 
\[
 \rho_S (\psi ) =\ketbra{\psi_1} + \cdots + \ketbra{\psi_k}, 
\]
\[
 \rho_S (\psi' ) =\ketbra{\psi_1'} + \cdots + \ketbra{\psi_k'}, 
\]
where vectors $\ket{\psi_i}$, respectively $\ket{\psi_i'}$, are orthogonal, i.e. $ \langle \psi_i  |  \psi_j \rangle = c_i \delta_{ij}$ and $ \langle \psi_i ' |  \psi_j '\rangle = c_i '\delta_{ij}$, then 
\begin{equation}
\label{eq8}
U_S  \ket{\psi_i} =\sum_{j=1}^k v_{ij}  \ket{\psi_j'} 
\end{equation}
for some elements $v_{ij}$, which form a unitary matric $V:= (v_{ij})$ iff all the vectors $\ket{\psi_i}$ and $\ket{\psi_i'}$ have the same norm.
\end{observation}

\begin{proof}
\cref{eq7} follows immediately from basic properties of partial trace and unitary operations. 

By imposing the orthogonality relations between vectors $\ket{\psi_i}$,  respectively $\ket{\psi_i'}$, one may extend both families (up to normalization of vectors $\ket{\psi_i}$ and $\ket{\psi_i'}$) into the basis of entire Hilbert space. 
Consider now the matrix $U_S$ in this basis. 
Since \cref{eq7} holds and $U_S$ is a unitary matrix, $U_S$ has a block structure, which transfers subspace spanned by vectors $\ket{\psi_i}$ onto the subspace spanned by vectors $\ket{\psi_i'}$, which implies \cref{eq8}. 

Assume now that vectors $\ket{\psi_i}$ and $\ket{\psi_i'}$ are normalized. 
Observe, that in the aforementioned basis, $V$ is a block matrix of $U_S$, and hence is a unitary matrix.
\end{proof}

Each AME($2k,d$) state $\ket{\psi}$ might be written in the following form
\[
\ket{\psi} = \sum_{I \in[d]^k} \ket{I} \otimes 
\ket{\phi_{I}} .
\]
where $I$ is multi-index $I=i_1 , \ldots , i_k $ which runs over the space $[d]^k$. 
If $\ket{\psi}$ is of minimal support, the vector $\ket{\phi_{I}} \in \mathcal{H}_d^{\otimes k}$  is separable in computational basis, i.e
\begin{equation}
\ket{\psi}  =\sum_{I \in[d]^k} \omega_I \ket{I} \otimes \ket{\phi_{I}^1} \otimes \cdots \otimes \ket{\phi_{I}^k},
\label{eq9}
\end{equation}
where vectors $\ket{\phi_{I}^j}$ are from the computational basis, i.e. $\ket{\phi_{I}^j} = \ket{0} ,\ldots ,\ket{d-1}$.

\begin{lemma}
\label{lem1}
Consider two AME states $\ket{\psi}$ and $\ket{\psi'}$ of the form
\begin{align*} 
\ket{\psi} &= \sum_{I =i_1,\ldots , i_k} \omega_I \ket{I} \otimes 
\ket{\phi_{I}} ,\\
\ket{\psi'} &= \sum_{I =i_1,\ldots , i_k} \omega_I' \ket{I} \otimes 
\ket{\phi_{I}'} .\\
\end{align*} 
which are local unitary equivalent by $U$. 
For any any multi-index $I=i_1 , \ldots , i_k $
\begin{align*}
U \Big( \omega_{I} \ket{\phi_{I}} \Big)=
\sum_{I' \in [d]^k} &
v^1_{i_1 i_1'} v^2_{i_1 i_1'} (i_1) \cdots v^k_{i_k i_k'} (i_1, \ldots, i_{k-1}) \\
&\omega_{I'}\ket{\phi_{I'}} ,
\end{align*}
where elements $v_{ij}^\ell (i_1, \ldots, i_{\ell-1}) $ forms a unitary matrices 
\[
V^\ell (i_1, \ldots, i_{\ell-1})  :=
 \Big(v_{ij}^\ell (i_1, \ldots, i_{\ell-1})  \Big)
 \]
for any $\ell =1,\ldots,k$ and indices $i_1, \ldots, i_{\ell-1}$. 
\end{lemma}

Notice that structure constants $v_{ij}^\ell (i_1, \ldots, i_{\ell-1}) $ depending on indices $i_1, \ldots, i_{\ell-1}$, which in fact is the main obstruction for obtaining more general results as those presented in this section. 

\begin{proof}
We shall use \cref{Ob1} repetitively $k$ times, by tracing out parties $1,\ldots,k$ respectively. In fact, the order of the procedure does not matter. In each step, the orthogonality of adequate vectors is fulfilled by relations $ \langle \phi_{I'} |  \phi_{I}\rangle = \delta_{I,I'}$. We present the first two step of the procedure in a more detailed way.

We know that $U \Big( \ket{\psi} \Big) = \ket{\psi'}$. Consider the partial traces over the first subsystem in both vectors $\ket{\psi}$ and $\ket{\psi'}$:
\begin{align*} 
\rho_{1^c} (\psi ) &= \sum_{i_1=0}^{d-1} \ketbra{\psi_{i_1}}, \\
\rho_{1^c} (\psi' ) &= \sum_{i_1=0}^{d-1} \ketbra{\psi_{i_1}'}, 
\end{align*} 
where
\begin{align} 
\label{eq11}
\ket{\psi_{i_1}} &= \sum_{I =i_2,\ldots, i_k} \omega_{i_1,I} \ket{I}\otimes \ket{\phi_{i_1, I}} ,
\end{align} 
and similarly:
\begin{align*} 
\ket{\psi_{i_1}'} &= \sum_{I =i_2,\ldots, i_k} \omega_{i_1,I }' \ket{I}\otimes \ket{\phi_{i_1, I}'} .
\end{align*}
From \cref{Ob1}, follows that
\begin{equation}
\label{eq10}
U_{1^c} \big(\ket{\psi_{i_1}} \Big) = 
\sum_{i_1'=0}^{d-1} v^1_{i_1 i_1'} \ket{\psi_{i_1'}'} ,
\end{equation}
and the elements $v^1_{i_1 i_1'}$ form the unitary matrix $V^1 := (v^1_{i_1 i_1'} )$.

We shall consider vectors $\ket{\psi_{i_1}} $ separately. 
For an arbitrary index ${i_1} $, consider a partial trace over second party of $\ket{\psi_{i_1}}$. From \cref{eq11}:
\begin{align*} 
\rho_{2^c} (\psi_{i_1} ) &= \sum_{i_2=0}^{d-1} \ketbra{\psi_{i_1, i_2}}, 
\end{align*} 
where
\begin{align*} 
\ket{\psi_{i_1,i_2}} &= \sum_{I =i_3,\ldots, i_k} \omega_{i_1,i_2,I} \ket{I}\otimes \ket{\phi_{i_1, i_2, I}}.
\end{align*} 
The partial trace over the second party of the right-hand side of \cref{eq10} is equal to:
\begin{equation*}
\sum_{i_2=0}^{d-1} \ketbra{\psi_{i_1, i_2}'}, 
\end{equation*}
where
\begin{equation*}
  \ket{\psi_{i_1, i_2}'}=  
\sum_{i_1'} v^1_{i_1 i_1'}   \sum_{I =i_3,\ldots, i_k} \omega_{i_1,i_2,I}'  \ket{I}\otimes \ket{\phi_{i_1, i_2 ,I}'} ,
\end{equation*}
and hence by \cref{Ob1} applied to the both sides of \cref{eq10}, we have:
\begin{equation}
\label{eq12}
U_{\{1,2\}^c} \Big(\ket{\psi_{i_1, i_2}} \Big) = 
\sum_{i_2'=0}^{d-1} v^2_{i_2 i_2'} (i_1) \ket{\psi_{i_1 ,i_2}'} ,
\end{equation}
for elements $v^2_{i_2 i_2'} (i_1)$, which form a unitary matrix.
Notice that the matrix $V^2 (i_1) := \big(v^2_{i_2 i_2'} (i_1) \big)$ is in particular dependent on the chosen index $i_1$.

We repeat presented procedure for arbitrary pair of indices $i_1$ and $i_2$; then $i_1$, $i_2$, and $i_3$; and in general $k$ times up to $i_1,\ldots,i_k$. 
Finally, we obtain
\begin{equation}
\label{eq120}
U_{\{1,\ldots,k \}^c} \Big(\ket{\psi_{i_1, \ldots,i_k}} \Big) = 
\sum_{i_k'=0}^{d-1} v^k_{i_k i_k'} (i_1 ,\ldots,i_{k-1}) \ket{\psi_{i_1 ,\ldots,i_{k}}'} 
\end{equation}
where elements $v^k_{i_k i_k'} (i_1,\ldots , i_{k-1})$, which form a unitary matrix. 
Notice that $\ket{\psi_{i_1, \ldots,i_k}} =\omega_I \ket{\phi_{i_1, \ldots,i_k}}$. 
On the other hand 
\begin{align*}
  \ket{\psi_{i_1, \ldots, i_k}'}= 
\sum_{i_1' ,\ldots ,i_{k-1}'} &
v^1_{i_1 i_1'} \cdots v^{k-1}_{i_{k-1} i_{k-1}'} (i_1, \ldots, i_{k-2}) \\
& \omega_{i_1,\ldots, i_{k}}'  \ket{\phi_{i_1, \ldots, i_k }'}.
\end{align*}
By the analysis of the recursion. Substitution of this formula to \cref{eq120} proves the proposition.
\end{proof}

\begin{corollary}
\label{cor1}
For two AME states of minimal support:
\begin{align*} 
\ket{\psi} &= \sum_{I =i_1,\ldots , i_k} \omega_I \ket{I} \otimes 
\ket{{\phi_{I}^1}}  \otimes  \cdots\otimes \ket{{\phi_{I}^k}}  ,\\
\ket{\psi'} &= \sum_{I =i_1,\ldots , i_k} \omega_I' \ket{I} \otimes 
\ket{{\phi_{I}^1}'}  \otimes  \cdots\otimes \ket{{\phi_{I}^k}'}  ,\\
\end{align*} 
which are local unitary equivalent by $U$, the equivalence is of the following form:  
\begin{align}
\label{eq13}
U_{[k]} \Big( \omega_{I}\ket{I}  \Big)=
\sum_{I'=i_1',\ldots, i_k'} &
v^1_{{\phi_{I}^1} {\phi_{I'}^1}'}  
v^2_{{\phi_{I}^2} {\phi_{I'}^2}'} ({\phi_{I}^1}) \cdots \\
& v^k_{{\phi_{I}^k} {\phi_{I'}^k}'} ({\phi_{I}^1}, \ldots ,{\phi_{I}^{k-1}})
\omega_{I'} \ket{I'} .\nonumber 
\end{align}
\end{corollary}

\begin{proof}
It follows immediately form \cref{lem1} applied to the second half of indices.
\end{proof}

So far obtained statement are rather technical. 
We shall demonstrate their effectiveness. 

\begin{proof}[Proof of \cref{cor2}]
Fix a multi-index $I=i_1, \ldots ,i_k$. By the definition,
\begin{equation}
\label{eq14}
U_{[k]} \Big( \ket{I}  \Big)= 
\sum_{I'=i_1',\ldots, i_k'} 
u^1_{i_1 i_1'} \cdots u^k_{i_k i_k'}
\ket{I'} .
\end{equation}
On the other hand, $U_{[k]} $ might be expressed in the form \cref{eq13}. 
For simplicity of the proof, we use the following notation:
\begin{align*}
u^\ell_{i_\ell'} :=& u^\ell_{i_\ell i_\ell'} \\
v^\ell_{{\phi_{I'}^\ell} '} :=&v^\ell_{\phi_{I}^\ell {\phi_{I'}^\ell} '}  
( \phi_I^1 ,\ldots , \phi_I^{\ell -1} )
\end{align*}
which is correctly defined once the index $I=i_1, \ldots ,i_k$ is fixed. 
Furthermore, we define:
\begin{align*}
\ell_{I'} :=& 
u^1_{i_1'}\cdots u^k_{i_k'} \\
p_{I'} :=&
v^1_{i_1'}\cdots v^k_{i_k'} ,
\end{align*}
for any multi-index $I'=i_1', \ldots ,i_k'$.

By comparison of \cref{eq14} and \cref{eq13}, we have:
\begin{equation}
\label{eq16}
 \ell_{I'} =p_{\phi_{I'} '}  \cdot\dfrac{\omega_{I'}'}{\omega_I},
\end{equation}
where $\phi_{I'} ':= ({\phi_{I'}^1}',\ldots, {\phi_{I'}^k}'  )$ for any $I'=i_1', \ldots ,i_k'$. 

Consider now the mutually orthogonal Latin hypercube MOLH($d$) (see \cref{hcc} for details):
\[
L(i_1' \ldots i_k') := \big({\phi_{I'}^1}',\ldots, {\phi_{I'}^k}' \big) .
\]
On the one hand, the set of indices $I'=i_1', \ldots ,i_k'$ for which $ \ell_{I'} \neq 0$ forms a rectangle $S:=S_1 \times \cdots \times S_k \subseteq [d]^k$. 
Indeed, it follows from the product form of  $\ell_{I'} $. 
On the other hand, the set of indices $I''=i_1'', \ldots ,i_k''$ for which $ p_{I''} \neq 0$ also forms a rectangle $S':=S_1' \times \cdots \times S_k' \subseteq [d]^k$, which follows from the product form of  $p_{I''} $. From equality in \cref{eq16} follows that $L$ maps $S$ onto $S'$. Form \cref{ob3}, $S$ and $S'$ are hypercubes and $L_{ \vert S}$ is a MOLH($s$), where 
$
s:= |S_1|=\cdots=|S_k|= |S_1'|=\cdots=|S_k'| .
$
This proves the second statement of \cref{cor2}. 

So far, we showed that for any fixed multi-index $I=i_1, \ldots ,i_k$, the relevant rows of matrices $U_1,\ldots,U_k$, i.e. vectors
\[
\big( u^1_{i_1 i} \big)_{i=0}^{d-1} , \ldots , \big( u^k_{i_k i} \big)_{i=0}^{d-1} 
\]
have the same number $s$ of non-zero elements. Observe, that it is also true for any other multi-index $I'=i_1', \ldots ,i_k'$. Indeed, the analogous reasoning, for other multi-index $I'_0=i_1, i_2', \ldots ,i_k'$ ($I'$ and $I'_0$ are equal on the first position) ensures us that the vectors
\[
\big( u^1_{i_1 i} \big)_{i=0}^{d-1} , \big( u^2_{i_2' i} \big)_{i=0}^{d-1} \ldots , \big( u^k_{i_k' i} \big)_{i=0}^{d-1}
\]
have the same number of non-zero elements, equal to $s$. 
From here, one can deduce it for arbitrary $I'=i_1', \ldots ,i_k'$. 
Since the inverse of unitary matrix is its conjugate transpose, the matrices $U_1^\dagger,\ldots,U_k^\dagger,\ldots$ provide the local unitary equivalence between $\ket{\psi'}$ and $\ket{\psi}$. 
Reasoning similar to the above proves that those matrices have the same number of non-zero elements in each row. This is equivalent to the fact that $U_1,\ldots,U_k$ have the same number of non-zero elements in each column. 
This proves the first statement of \cref{cor2} for matrices $U_1,\ldots,U_k$ (without stating the equality of element's norms). 
Observe, that by taking another set of indices $S \subset [n]$ one may extend this reasoning to all matrices $U_1,\ldots,U_k, \ldots, U_{2k}$. 

In fact, more detailed analysis of matrices $U_1,\ldots,U_k$ might be performed. 
Once more, fix $I=i_1, \ldots ,i_k$, and keep the notation introduced before in the proof. 
For each $\bar{I}'=i_1' ,\ldots , i_{k-1}' \in S_1 \times \cdots\times S_{k-1}$, we have:
\begin{equation}
\label{eq20}
\prod_{i_k' \in S_k} \ell_{\bar{I}' i_k'} = \big(u^1_{i_1'}\big)^s \cdots \big(u^{k-1}_{i_{k-1}'}\big)^s C_1
\end{equation}
where $C_1=\prod_{i_k' \in S_k} u^k_{i_k'}$. 
By \cref{eq16}, the left-hand side of \cref{eq20} is equal to
\begin{equation}
\label{eq21}
\prod_{i_k' \in S_k} p_{\phi_{\bar{I}' i_k'} '} 
\cdot \dfrac{\omega_{\bar{I}' i_k'}'}{\omega_I} =
\dfrac{1}{\omega_I} W^{[k-1]}_{\bar{I}'} 
\prod_{i_k' \in S_k} 
v^1_{L^1_{\bar{I}' i_k'}} \cdots v^k_{L^k_{\bar{I}' i_k'}} .
\end{equation}
Here, we use the notation of $W^{[k-1]}_{\bar{I}'}$ introduced in \cref{IIIA} with the slight modification. Namely, the product:
\[
W^{[k-1]}_{\bar{I}'}:= 
\prod_{i_k' \in S_k} \omega_{\bar{I}' i_k'}
\] 
runs only over all non-zero elements $\omega_{\bar{I}' i_k'}$, which is exactly $s$. 
One of the basic properties of MOLH is that in each row and on each position, all elements appear exactly once (see \cref{MOLH} for details). Hence, \cref{eq21} is equal to
\begin{equation}
\label{eq22}
\dfrac{1}{\omega_I} W^{[k-1]}_{\bar{I}'} 
C_2, \quad\quad\text{where}\quad 
C_2 =
\prod_{I'' \in S'} 
v^1_{i_1''} \cdots v^k_{i_k''} .
\end{equation}
\cref{eq20,eq21,eq22} combines to the following:
\begin{equation}
\label{eq23}
\big(u^1_{i_1'}\big)^s \cdots \big(u^{k-1}_{i_{k-1}'}\big)^s
\dfrac{1}{W^{[k-1]}_{\bar{I}'} }
=
\dfrac{1}{\omega_I} 
\dfrac{C_2}{C_1}
\end{equation}
for each $\bar{I}'=i_1' ,\ldots , i_{k-1}' \in S_1 \times \cdots\times S_{k-1}$, where constants $C_1$ and $C_2$ are independent on $\bar{I}'$. 
From \cref{eq23}, one can deduce the proportions of non-zero elements in rows of $U_1$. 
Indeed, choose a multi-index $\widetilde{I}' \in S_2 \times \cdots \times S_{k-1}$ and two indices $i_1' , i_1''\in S_1$. From \cref{eq23} applied to $\bar{I}' =i_1' \widetilde{I}'$ and $\bar{I}' =i_1'' \widetilde{I}'$ follows that:
\begin{equation}
\label{eq24}
\big(u^1_{i_1'}\big)^s 
\dfrac{1}{W^{[k-1]}_{i_1' \widetilde{I}'} }
=
\big(u^1_{i_1''}\big)^s 
\dfrac{1}{W^{[k-1]}_{i_1'' \widetilde{I}'} } .
\end{equation}
Since $W^{[k-1]}_{i_1' \widetilde{I}'}$ and $W^{[k-1]}_{i_1'' \widetilde{I}'}$ have the same norms, there is following equality: $|u^1_{i_1'}|=|u^1_{i_1''}|$. 
Those, however, under introduced notation denote elements of local unitary matrix: $u_{i_1 ,i_1'}$ and $u_{i_1 ,i_1''}$. 
We conclude that in $i_1$th row of the matrix $U_1$ there are exactly $s$ non-zero elements, all having the same norm. 
Since $U_1$ is a unitary matrix, in particular, it preserves the norms, all non-zero elements form the $i_1$th row have norm equal to $1/\sqrt{s}$. 
Analogous reasoning might be performed with respect to each row of matrix $U_1$ and further to each matrix $U_i$, $i=1,\ldots,2k$. 
Therefore, we conclude the equality of element's norms in the first statement of \cref{cor2}.

In fact, \cref{eq24} is the last general result concerning the characteristic of local unitary matrices $U_i$. 
We restrict now to the case when $\omega_I \equiv \omega_I ' \equiv1$ for all multi-indices $I=i_1, \ldots i_k$. 
Obviously, $W^{[k-1]}_{i_1' \widetilde{I}'} =1 $ for all multi-indices $i_1' \widetilde{I}' \in S_1 \times \cdots\times S_{k-1}$. 
Thus, from \cref{eq24} follows that all non-zero entries of $i_1$th row have not only equal norms but also equal $s$th powers. 
In other words, all non-zero entries of $i_1$th row are $s$th roots of unity up to some scaling complex number $w_{i_1}$. 
Analogous reasoning might be performed with respect to each row of matrix $U_1$. 
In such a way we obtain scaling factors $w_{i}$ for $i=1,\ldots,d$. 
Observe that matrix $U_1$ multiplied by the diagonal matrix
\[
\sqrt{s} \; \text{diag} \Big( w_1 ,\ldots  ,w_d \Big)
\]
consists only of zeros and $s$th roots of unity. Similarly, one can show the same property for local matrices $U_2,\ldots,U_{2k}$. 
This proves the third statement of \cref{cor2}.
\end{proof}

We conjecture that the matrices $U_1,\ldots,U_k, \ldots, U_{2k}$ from \cref{cor2} have the block structure:
\[
U_\ell=
\begin{bmatrix}
S_{\ell, 1} & 0 & \cdots& 0\\
0& S_{\ell,2} &  \cdots& 0\\
\vdots & \vdots & \ddots & \vdots\\
0& 0&  \cdots& S_{\ell,d/s}
\end{bmatrix}
,
\]
where matrices $S_{i,j} $ are $s \times s$ unitary matrices with all entries of the same norm. 
Unfortunately, coefficients $v_{ij}^\ell (i_1, \ldots, i_{\ell-1}) $ depend on indices $i_1, \ldots, i_{\ell-1}$, which impose selection of the index $I$ at the beginning of the proof above. 
This enables us to deduce the general block structure of matrices $U_i$. 
Notice that not all unitary matrices having the same number of non-zero elements in each row and column of the same norm each are necessarily of the block structure.

Concluding the block structure of matrices $U_\ell$ is is not within our reach yet. Nevertheless, in two specific cases when $s=1$ or $s=d$ the block structure is obvious. 
Interestingly, those two values of $s$ are the only possibilities for most of small dimensional AME(2k,d) states, see \cref{small}. 
This follows from the requirement of appropriate dimensional MOLH extension. 

\begin{proof}[Proof of \cref{prop1=,verification2}]
Suppose that the matrix $U:=U_1 \otimes\ldots\otimes U_{2k}$ provides an LU-equivalence between two AME(2k,d) states with minimal support. 
Assumption on $k$ and $d$ being small enough are equivalent to the fact that the only possible values of $s$ are simply $1$ and $d$. 

In the first case when, $s=1$, matrices $U_i$ are, by the definition, monomial matrices.
This gives the second possibility in \cref{prop1=}. 
Each monomial matrix is a product of permutation and diagonal matrix, hence the form of \cref{AMEeq2} in \cref{verification2} follows. 
Similarly to the proof of \cref{verification} presented in \cref{verification2APP}, it may be shown that the form of diagonal matrices $D_i$ from \cref{AMEeq2} are exactly as it is indicated in \cref{verification2}.

The second case, when $s=d$, is new and goes beyond the analysis performed so far. 
We investigate matrix $U_1$.
Let us recall the last general formula in our analysis, namely \cref{eq24}. 
By descrambling the notation introduced in the proof of \cref{cor2}, \cref{eq24} takes the following form:
\begin{equation}
\label{Descrambling}
\big(u^1_{i j}\big)^d 
\dfrac{1}{W^{[k-1]}_{j  I} }
=
\big(u^1_{i j'}\big)^d 
\dfrac{1}{W^{[k-1]}_{j'I} } ,
\end{equation}
where $i\in [d]$ and $ j,j'\in [d]$ are arbitrary indices, and $I\in [d]^{k-2}$ is arbitrary multi-index. 
This equation describes the proportion of $d$th powers of elements in $i$th row of the matrix $U_1$. 
Observe that they are independent on the index $i$. 
Indeed, multiply the matrix $U_1$ on the right-hand side by the following diagonal matrix: 
\[
 \overleftarrow{D}_1=
\text{diag} \Bigg( \sqrt[\leftroot{-3}\uproot{3}d]{(W^{[k-1]}_{0,I})},
\ldots ,
\sqrt[\leftroot{-3}\uproot{3}d]{(W^{[k-1]}_{d-1,I})} \Bigg) .
\]
Observe that the entries of $ \bar{U}_1 := U_1 \overleftarrow{D}_1 $ satisfies 
\begin{equation*}
\big(\bar{u}^1_{i j}\big)^d 
=
\big(\bar{u}^1_{i j'}\big)^d ,
\end{equation*}
for any number $i$ indexing the rows and any pair of indices $j,j' \in [d]$. 
As we already discussed while proving \cref{cor2}, the hermitian-conjugate matrix $U^\dagger$ provides the reverse LU-equivalence. 
One can analyze the proportions of elements in rows of $U^\dagger_1$ in the same way as we did for $U_1$. 
The analogue of \cref{Descrambling} yields the following conclusions on the columns of matrix $U_1$:
\begin{equation}
\label{Descrambling2}
\big(u^1_{i j}\big)^d 
\dfrac{1}{(W^{[k-1]}_{i I} )'}
=
\big(u^1_{i' j}\big)^d 
\dfrac{1}{(W^{[k-1]}_{i'I} )'} ,
\end{equation}
where $i,i', j\in [d]$ are arbitrary indices; and $I\in [d]^{k-2}$ is arbitrary multi-index. 
Therefore, by multiplying the matrix $\bar{U}_1$ from the left-hand side by the diagonal matrix
\[
 \overrightarrow{D}_i =
\text{diag} \Bigg( \sqrt[\leftroot{-3}\uproot{3}d]{(W^{[k-1]}_{0,I})'},
\ldots ,
\sqrt[\leftroot{-3}\uproot{3}d]{(W^{[k-1]}_{d-1,I})'} \Bigg) ,
\]
we obtain the matrix
\[
\widetilde{U}_1 :=  \overrightarrow{D}_i U_1  \overleftarrow{D}_1
\]
with the following property:
\[
\big(\widetilde{u}^1_{i j}\big)^d 
=
\big(\widetilde{u}^1_{i' j'}\big)^d ,
\]
for any indices $i,i'j,j' \in [d]$. 
Up to some global factor $\omega_1$, all entries of the matrix $\widetilde{U}_1 $ are $d$th roots of unity. 
By the definition $\omega_1 \widetilde{U}_1 $ is a Butson-type matrix. 

We have shown that under the assumption $s=d$, the statement in \cref{AMEeq1} of \cref{verification2} holds for the matrix $U_1$ and the set $S $ of the consecutive next $k-2$ indices. 
The same reasoning might be applied for other matrix $U_i$ and set $S$. 
The global phase in \cref{AMEeq1} is then a multiplication of obtained factors $\omega_i$. 

What is left for the analysis is the necessary condition for existence of LU-equivalence.
Consider \cref{Descrambling}. 
Since $I\in [d]^{k-2}$ is arbitrary multi-index and the ratio of two matrix elements: $u^1_{i j} /u^1_{i j'} $ is constant, we immediately obtain:
\[
\dfrac{W^{[k-1]}_{jI}}{W^{[k-1]}_{j'I}} =
\dfrac{W^{[k-1]}_{jI'}}{W^{[k-1]}_{j'I'}}
\] 
for any multi-indices $I,I'\in [d]^{k-2}$. 
Similarly, from \cref{Descrambling2}, we have:
\[
\dfrac{(W^{[k-1]}_{jI})'}{(W^{[k-1]}_{j'I})'} =
\dfrac{(W^{[k-1]}_{jI'})'}{(W^{[k-1]}_{j'I'})'} .
\]
This ends the proof of \cref{verification2}. 
Obviously, \cref{prop1=} is an immediate consequence of \cref{verification2}.
\end{proof}

\section{The proof of \cref{AME55}}
\label{AppD}

We begin this section manifesting the problem of LU- verification of given two AME states. 
AME states are $k$-uniform states which saturates the singleton bound on the uniformity $k$. 
In general, for given two $k$-uniform states, one can compare ranks of reduced density matrices in order to exclude a local equivalence between \cite{RanksReducedMatrices}. 
We illustrate  this phenomenon on the following example of two $1$-uniform states of four qubits:
\[
\Ket{\psi_1} =\dfrac{1}{\sqrt{2}}\sum_{i=0}^{1} \Ket{i,i,i,i}, \quad 
\Ket{\psi_2} =\dfrac{1}{2}\sum_{i,k=0}^{1} \Ket{i,k,k,i+k} .
\]
Observe that $\text{rank} \rho_{12} ( \psi_1 )\neq \text{rank} \rho_{12} ( \psi_2 ) $, hence the states $\Ket{\psi_1}$ and $\Ket{\psi_1}$ are not LU-equivalent. 
Nevertheless, this simple argument is never conclusive for both states being AME. 
Indeed, all reduced density matrices of AME state $\ket{\psi}$ have precisely determined ranks: $\text{rank} \rho_S (\psi ) =\text{min}\{|S| ,|S^c|\}$.

Our initial attempt for showing that states $\ket{\text{AME(5,d)}}$ and $\ket{\text{AME(5,d)'}}$
presented in \cref{ex3} and \cref{ex5} are not locally equivalent was to reduce this problem to smaller subsystems. 
We investigated the reduced density matrices: 
\begin{align*}
\rho_{345} \Big( &\text{AME(5,d)'}\Big)   \nonumber \\
&= \sum_{i,j=0}^{d-1} \ketbra{i+j, i+2j, i+3j}
\end{align*}
\noindent
and
\begin{align*}
\rho_{345} \Big( &\text{AME(5,d)} \Big) \\
&= \sum_{i,j=0}^{d-1}  \Bigg( \sum_{k,k'=0}^{d-1}  
\omega^{(i +3j )(k-k')}
 \ket{i+j, k +i+2j,  k}  \\
&  \;\;\;\,\,\,\,\,\,\,\,\,\,\,\,\,\,\,\,\quad\quad\quad
\bra{i+j, k' +i+2j, k'} \Bigg) ,
\end{align*}
where $\omega$ is $d$th root of unity. 
Even though $\rho_{345}( \text{AME(5,d)'} )$ and $\rho_{345} ( \text{AME(5,d)}) $ have the same rank, we attempt to show that they are not LU-equivalent. 
Surprisingly, it turned out that they are. 
We present both reduced density matrices for a local dimension $d=3$ on \cref{fig1} and \cref{fig2} respectively.

\begin{figure}

\caption{\label{fig2} The density matrix $ \rho_{345} \left( \text{AME(5,d)'} \right) $.}
\end{figure}

\begin{figure}
%
\caption{\label{fig1} The density matrix $ \rho_{345} \left( \text{AME(5,d) }\right) $.}
\end{figure}

It is not a straightforward observation that both states are actually LU-equivalent. 
Standard procedures of diagonalizing the density matrix $ \rho_{345} \left( \text{AME(5,d)} \right) $ lead to non-local operations. 
Attentive analysis of low local dimensions ($d=3,5$) showed that diagonalization might be performed in a local way. 
Indeed, for a local dimension $d=3$, the following unitary matrices: $\widetilde{U}_3 =\Id_3$,
\[
 \widetilde{U}_4 = \dfrac{1}{3}
  %
 .
\]
provides a local equivalence between $\rho_{345}( \text{AME(5,3)'} )$ and $\rho_{345} ( \text{AME(5,3)}) $. 
This equivalence is illustrated on \cref{fig3}. 
Similarly, for a local dimension $d=5$, matrices: $\widetilde{U}_3 =\Id_5$,
\[
\widetilde{U}_4 = \dfrac{1}{5}
  %
 .
\]
provides a local equivalence between $\rho_{345}( \text{AME(5,5)'} )$ and $\rho_{345} ( \text{AME(5,5)}) $. 
In order to provide the general formula, for any odd local dimension $d$, we introduce the following recursive construction of two $(d \times d)$-dimensional matrices:
\[
W=\left( w_{ij}\right)\text{  and  } 
V =\left( v_{ij}\right)
\]
with coefficients in GF($d$):
\begin{enumerate}[(a)]
\item $ w_{00} =0$, $v_{00} =0$;
\item $w_{0(j+1)}$ is defined by the formula: $ w_{0 (j+1)} =w_{0j} +2j$; similarly, $v_{0 (j+1)} =v_{0j} -2j$ (definition of firsts rows);
\item $ w_{ij} :=  w_{(i-1) (j-1)}$, similarly, $v_{ij} := {v}_{(i-2) (j-1)}$ (definition of succeeding rows).
\end{enumerate}
Notice that in order to define rows of the matrix $V$ correctly, we impose $2 \nmid d$. 
The matrices $\widetilde{U}_4$ and $\widetilde{U}_5 $ are of the following form
\[
\widetilde{U}_4 =\left( \omega^{w_{ij}} \right)  \text{  and  } 
\widetilde{U}_5 =\left( \omega^{v_{ij}}\right) .
\]
This construction overlap with the aforementioned constructions in dimensions $d=3,5$. 
In fact, there is a close formula for entries of matrices $W$ and $V$ (and hence for matrices $\widetilde{U}_4$ and $\widetilde{U}_5$) given in terms of triangular numbers:
\[
{w}_{ij} = 2t_{j-i-1}, \quad {v}_{ij} =
\left\lbrace 
\begin{matrix}
-2t_{j-i/2-1} & \text{ for } 2\mid i\\
-2t_{j-(i+d)/2-1} & \text{ for } 2\nmid i\\
\end{matrix}
\right. 
\]
where $t_i = 0, 1, 3, 6, 10, 15,\ldots $ are consecutive triangular numbers defined as
\begin{equation}
\label{formula}
t_k =\sum_{i=0}^k i =\dfrac{(k+1)k}{2}.
\end{equation}

\begin{lemma}
\label{LastLEmma}
For any odd local dimension $d$, the matrices $\widetilde{U}_4$ and $\widetilde{U}_5 $ provide the LU-equivalence between $\rho_{345}( \text{AME(5,d)'} )$ and $\rho_{345} ( \text{AME(5,d)}) $, i.e.
\[
\rho_{345}( \text{AME(5,d)'} ) = \Id \otimes \widetilde{U}_4\otimes \widetilde{U}_5 
\Bigg(\rho_{345} ( \text{AME(5,d)}) \Bigg).
\]
\end{lemma}

\begin{proof}
Observe that:
\begin{align*}
\sum_{s,j =0}^{d-1} &\ket{s,s+j,s+2j}\bra{s,s+j,s+2j} \xmapsto{\Id_3 \otimes\widetilde{U}_4 \otimes\widetilde{U}_5 } \\
&\begin{aligned} \sum_{s=0}^{d-1}\sum_{\substack{m,m',k, \\ k' =0}}^{d-1} 
\Big( \sum_{j=0}^{d-1} &\omega^{\big(
w_{(s+j) m} \overline{w_{(s+j) m'}} 
v_{(s+2j) k} \overline{v_{(s+2j) k'}} \big)} \\
&\ket{s,m,k}\bra{s,m',k'} \Big) .
\end{aligned}\\
\end{align*}
We examine the coefficient by 
\[
\ket{s,k+s+j,k}\bra{s,k'+s +j',k'} 
\]
in the expression above:
\begin{align}
\label{xxx}
\dfrac{1}{d}\sum_{i=0}^{d-1} &
\omega^{\big( 
w_{(s+i) (k+j+s)} \overline{w}_{(s+i) (k'+j'+s)} 
v_{(s+2i) k} \overline{v}_{(s+2i) k'} \big)}=\nonumber\\
&\dfrac{1}{d}\sum_{i=0}^{d-1} 
\omega^{2 \big( 
t_{k+i-j-1} -t_{k'+j'-i-1} -
t_{k-s/2-i-1} t_{k' -s/2-i-1} \big)} ,
\end{align}
where we assumed $ 2\mid s$; argument for the opposite case, where $ 2\nmid s$, is similar to the one we present. Using \cref{formula}, after elementary transformations \cref{xxx} is equal to:
\begin{align*}
\dfrac{1}{d} &\omega^{\big( s(k-k') +2jk-2j'k' \big)}
\omega^{\big( j(j-1) -j'(j'-1) \big)}
\underbrace{\sum_{i=0}^{d-1} \omega^{\big( 2i(j-j')\big)}}_{d \delta_{j j'}} =\\
&\omega^{(s+2j)(k-k')},
\end{align*}
what remained to be shown.
\end{proof}

\begin{figure}

\caption{\label{fig3} The matrix $\Id \otimes\widetilde{U_4}\otimes\widetilde{U_4}$ transforms $\rho_{cde} \left( \psi \right)$ onto $\rho_{cde} \left( \phi \right)$. The figure above presents how does it act on the first block of \cref{fig2}.}
\end{figure}

\begin{proof}[Proof of \cref{AME55}]
We show the statement by contradiction. 
Assume that states $\ket{\text{AME(5,d)}}$ and $\ket{\text{AME(5,d)'}}$ are LU-equivalent by some unitary matrices: 
\begin{equation*}
U_1 \otimes U_2 \otimes U_3 \otimes U_4 \otimes U_5=: U_{12} \otimes U_{345}. 
\end{equation*}
We keep "$\otimes$" in our notation in order to distinguish it from the matrix multiplication. 
Since the (partial) trace is invariant under cycling permutations, we have:
\begin{align*}
&\rho_{345} \big( \text{AME(5,d)'}\big) \nonumber \\
&=\tr_{12}  \ketbra{\text{AME(5,d)'}}  \nonumber \\
&=\tr_{12} \left( U_{12}\otimes U_{345} \ketbra{\text{AME(5,d)}}  U_{12}^{-1} \otimes U_{345}^{-1} \right)   \nonumber \\
&= U_{345} \Big(
\tr_{12} \left(U_{12} \ketbra{\text{AME(5,d)}}  U_{12}^{-1} \right) 
\Big) U_{345}^{-1}  \nonumber \\
&= U_{345} \Big(
\tr_{12} \left( \ketbra{\text{AME(5,d)}}   \right) 
\Big) U_{345}^{-1}    \nonumber \\
&= U_{345} \Big(
\rho_{345} \left( \text{AME(5,d)}\right)
\Big) U_{345}^{-1} .  \nonumber \\
\end{align*}
\noindent
Hence, the operator $ U_{345} := U_3 \otimes U_4 \otimes U_5$ provides the local equivalence between $\rho_{345} (\text{AME(5,d)} )$ and $\rho_{345} (\text{AME(5,d)'})$. 

Notice that in \cref{LastLEmma}, we pointed out that the earlier constructed matrices $\widetilde{U}_4$ and $\widetilde{U}_5 $ provide the LU-equivalence between $\rho_{345}( \text{AME(5,d)'} )$ and $\rho_{345} ( \text{AME(5,d)}) $ for any odd local dimension $d$. 
Precisely:
\[
\rho_{345}( \text{AME(5,d)} ) = \Id \otimes \widetilde{U}_4\otimes \widetilde{U}_5 
\Bigg(\rho_{345} ( \text{AME(5,d)'}) \Bigg).
\]
Therefore from \cref{prop1} we conclude that 
\[
 U_3 =  M_3  ,\quad
 U_4=  M_4 \widetilde{U}_4 ,\quad
U_5 =M_5  \widetilde{U}_5 ,
\]
for some monomial matrices $M_3, M_4, M_5$. 
Indeed, $U_{345} (\Id \otimes \widetilde{U}_4\otimes \widetilde{U})^{-1}$ constitutes an automorphism of $\rho_{345} ( \text{AME(5,d)'}) $, and hence by \cref{prop1}, is a tensor product of monomial matrices. 
We shall prove that such restriction on matrices $U_3,U_4,U_5$ leads to a contradiction.

To sum up the discussion so far, LU-equivalence between states $\ket{\text{AME(5,d)}}$ and $\ket{\text{AME(5,d)'}}$ has the following form:
\begin{equation*}
U_1 \otimes U_2  \otimes M_3 \otimes M_4 \widetilde{U}_4 \otimes M_5 \widetilde{U}_5
\end{equation*}
where $U_i$ are arbitrary unitary matrices, while $M_i$ are product of diagonal and permutation matrices. 
Therefore
\begin{equation}
\label{star}
\ket{\text{AME(5,d)}} = \Big( U_1 \otimes U_2  \Big)  \ket{i,j}  \otimes M_3 \ket{i+j} \otimes  B_{ij}
\end{equation}
where
\[
B_{ij} := \Big( M_4 \widetilde{U}_4 \otimes M_5 \widetilde{U}_5 \Big)   
 \ket{i+2j,i+3j}
\]
Observe that $B _{ij}$ are linearly independent. Indeed, they are unitary transformed linearly independent vectors $\ket{i+2j,i+3j}$. 

We shall show that matrices $U_1,U_2$ are monomial matrices. 
Suppose for simplicity, that $M_3 =\Id$. 
Recall that $\ket{\text{AME(5,d)}}$ has the following form:
\begin{equation}
\label{aux2}
\ket{\text{AME(5,d)}} =  \Ket{i,j,i+j} \otimes   C_{ij} ,
\end{equation}
where
\[
C_{ij} = \omega^{(i+3j)k}  \Ket{k+2j,k} .
\]
We compare this expression with \cref{star}. 
Suppose now, that in some column of the matrix $U_1$, there are at least two non-zero elements: $u_{lk}^1, u_{l'k}^1$ ($l\neq l'$); consider some non-zero element $u_{nm}^2$ of the matrix $U_2$. 
Observe, that it leads to the following expressions 
\[
\ket{l,n} \otimes \ket{k+m} \otimes B_{km}  \quad\text{and} \quad
\ket{l',n} \otimes \ket{k+m} \otimes B_{km}
\]
in \cref{star}. 
Clearly, there is an additional contribution from other non-zero elements of matrices $U_1$ and $U_2$. 
Since $B_{ij}$ are linearly independent, there are the following terms
\[
\ket{l,n} \otimes \ket{k+m} \otimes D_{ln}  \quad\text{and} \quad
\ket{l',n} \otimes \ket{k+m} \otimes D_{l'n}
\]
in \cref{star}, where $D_{ln}$ and $D_{l'n}$ are some non-zero elements. 
Observe, that such terms might appears in \cref{aux2} only if $l+n =k+m$ and $l'+n =k+m$, which is contradictory to $l\neq l'$.
Similarly, one can show that none of the matrix $U_1$ columns have two non-zero elements. 
We have shown that, indeed, matrices $U_1$ and $U_2$ are monomial under the assumption $M_3 =\Id$. 
Nevertheless, the assumption $M_3 =\Id$ is not essential here, the similar argument might be given for arbitrary monomial matrix $M_3$. 
Hence $U_1$ and $U_2$ are monomial matrices in general. 

Observe that $\text{supp} (B_{ij} )=d^2$. Indeed, 
\[
\text{supp} \Big(\widetilde{U}_4 \otimes \widetilde{U}_5 \ket{i+2j,i+3j} \big) =d^2 ,
\]
and the monomial operators $M_4 ,M_5$ do not change the support. 
Since $U_1$ and $ U_2 $ are monomial matrices, the support of the right-hand side in \cref{star} is equal to $d^4$. 
This is contradictory to the fact that $ \text{supp} (\ket{\text{AME(5,d)}}  )=d^3$.
\end{proof}

We have shown that two families of AME(5,d) states are not LU-/SLOCC-equivalent. 
Even though only a special family of states is considered here, analysis of the proof of \cref{AME55} reveals the general method of LU-verification of AME and $k$-uniform states where one of them is written with minimal support. 
Firstly, the formula for LU-equivalence between reduced $(k+1)$-dimensional systems should be provided. 
Secondly, based on \cref{prop1}, one can classify such equivalences between those reduced states. 
Finally, it should be show that none of such equivalences can be extended to the local equivalence of initial states.

\bibliography{Physics.bib}
\end{document}